\documentclass[12pt, english]{article}

\usepackage{comment}

\usepackage[T1]{fontenc}
\usepackage[latin9]{inputenc}
\usepackage{geometry}
\usepackage[colorinlistoftodos,prependcaption,textsize=tiny,draft,shadow,textwidth=2.5cm]{todonotes}

\usepackage[cal=boondoxo,frak=euler ,scr=euler,bb= pazo]{mathalfa}

\usepackage{tikz}
\usetikzlibrary{shapes.geometric, arrows, positioning, fit, backgrounds, patterns}

\usepackage{esstixcal}

\usepackage{mathtools}
\usepackage{nicefrac,xfrac}
\makeatletter

\usepackage{relsize}
\usepackage{commath}
\usepackage{upgreek}
\usepackage{amsmath}
\usepackage{amsthm}

\usepackage{amssymb}
\usepackage{setspace}
\usepackage{enumitem}
\usepackage{nicefrac}
\usepackage{bibentry}
\usepackage[authoryear]{natbib}
\PassOptionsToPackage{normalem}{ulem}
\usepackage{ulem}

\usepackage{enumitem}
\makeatletter

\usepackage{bold-extra}

\usepackage{lipsum}                     
\usepackage{xargs}    
\usepackage{mathtools}
\usepackage{hyperref}
\usepackage{cancel}
\usepackage[normalem]{ulem}
 \hypersetup{
breaklinks=true,
      colorlinks   = true,
      citecolor    = black,
 urlcolor=black,
 linkcolor=black,
 }

\usepackage[multiple]{footmisc}

\makeatother

\usepackage{babel}

\usepackage{fourier}
\usepackage{microtype}

\DeclareSymbolFont{CMsymbols}{OMS}{cmsy}{m}{n}
\SetSymbolFont{CMsymbols}{bold}{OMS}{cmsy}{b}{n}
\DeclareMathSymbol{\sim}{\mathrel}{CMsymbols}{"18}

\DeclareSymbolFont{CMsymbols}{OMS}{cmsy}{m}{n}
\SetSymbolFont{CMsymbols}{bold}{OMS}{cmsy}{b}{n}
\DeclareMathSymbol{\succ}{\mathrel}{CMsymbols}{"1F}

\usepackage{bold-extra}

\newtheoremstyle{faber}%
{8pt}{8pt}%
{\itshape}{}%
{\scshape}{.}%
{ .5em}%
{\thmname{#1}\thmnumber{ #2}\thmnote{{\normalfont\ (#3)}}}

\newtheoremstyle{faber_ex}%
{8pt}{8pt}%
{\normalfont}{}%
{\scshape}{.}%
{ .5em}%
{\thmname{#1}\thmnumber{ #2}\thmnote{{\normalfont\ (#3)}}}

\theoremstyle{faber}
\newtheorem{theorem}{Theorem}

\newtheorem{proposition}{Proposition}

\newtheorem{lemma}{Lemma}

\newtheorem*{axiom}{Axiom}

\newtheorem{corollary}{Corollary}

\theoremstyle{faber_ex}
\newtheorem{remark}{Remark}
\newtheorem{example}{Example}

\definecolor{airforceblue}{rgb}{0.36, 0.54, 0.66}
\definecolor{amber}{rgb}{1.0, 0.49, 0.0}
\definecolor{amethyst}{rgb}{0.6, 0.4, 0.8}

\definecolor{ao}{rgb}{0.0, 0.0, 1.0}
\definecolor{amaranth}{rgb}{0.9, 0.17, 0.31}
\definecolor{antiquefuchsia}{rgb}{0.57, 0.36, 0.51}

\newcommand{\faberd}[1]{}

\newcommand{\virgo}[1]{``#1''}

\newcommand{\R}{\mathbb{R}}
\newcommand{\B}{B_0(\Sigma)}
\newcommand{\BK}{B_0(\Sigma,K)}

\newcommand{\F}{\mathcal{F}}

\title{Absolute and Relative Ambiguity Attitudes\thanks{We are grateful to Xiaoyu Cheng, Federico Echenique, Mira Frick, Bruno de Albuquerque Furtado, Itzhak Gilboa, Paolo Ghirardato, Simon Grant, Faruk Gul, Brian Hill, Ryota Iijima, Peter Klibanoff, Massimo Marinacci, Stefania Minardi, Ivan Moscati, Sujoy Mukerji, Pietro Ortoleva, Daniele Pennesi, and Luciano Pomatto for their insightful comments. We thank the audience at DTEA 2024. We are particularly indebted to Simone Cerreia-Vioglio and Fabio Maccheroni for many discussions and suggestions.}} 

\author{\href{https://www.francesco-fabbri.com}{Francesco Fabbri}\thanks{Princeton University, Department of Economics, \url{ffabbri@princeton.edu}.} $\quad$  \href{https://sites.google.com/view/giulioprincipi/home-page}{Giulio Principi}\thanks{New York University, Department of Economics, \url{gp2187@nyu.edu}.} $\quad$ \href{https://lorenzomstanca.com/}{Lorenzo Stanca}\thanks{\hbox{Collegio Carlo Alberto and University of Turin, ESOMAS, \url{lorenzo.stanca@carloalberto.org}.}}}

\date{\today}

\begin{document}

\maketitle

\begin{abstract}
\noindent
We represent preferences that exhibit absolute or relative attitudes towards ambiguity without assuming convexity of preferences. Our analysis is motivated by the recent experimental evidence by \cite{BaillonTestincCAAARA}
indicating that ambiguity becomes more tolerable as individuals' welfare improves. Decreasing absolute ambiguity aversion is characterized by constant superadditive certainty equivalents and admits an act-dependent variational representation (\citeauthor{Variational}, \citeyear{Variational}). Decreasing relative ambiguity aversion relates to positive superhomogeneity   and admits an act-dependent confidence preference representation (\citeauthor{chateauneuf2009ambiguity}, \citeyear{chateauneuf2009ambiguity}). We apply our characterizations to retrieve a classic risk sharing result on the efficiency of trade and subjective beliefs of the individuals (\citeauthor{rigotti2008subjective}, \citeyear{rigotti2008subjective}).

\end{abstract}

\vspace{1cm}

\vspace{0.2cm}


\newpage

\section{Introduction}

The seminal work of \cite{ellsberg1961risk} originated a vast theoretical literature on ambiguity that successfully addressed several decision puzzles in the context of uncertainty. In parallel, a growing literature has employed these decision-making models to study a wide range of economic applications, e.g., risk sharing (\citeauthor{rigotti2008subjective}, \citeyear{rigotti2008subjective}; \citeauthor{ghirardato2018risk}, \citeyear{ghirardato2018risk}); moral hazard (\citeauthor{miao2016robust}, \citeyear{miao2016robust}); portfolio choice (\citeauthor{maccheroni2013alpha}, \citeyear{maccheroni2013alpha}). However, many popular ambiguity models fail to comply with the experimental evidence concerning how attitudes towards ambiguity change when individual welfare changes (\citeauthor{BaillonTestincCAAARA}, \citeyear{BaillonTestincCAAARA}). This paper attempts to fill this gap by providing general preference representations for absolute and relative ambiguity attitudes.

Ambiguity models often assume independence notions which limit their ability to capture the whole spectrum of ambiguity attitudes (see Figure \ref{fig:MBA_changing}). For instance, the class of invariant biseparable preferences (\citeauthor{ghirardato2004differentiating}, \citeyear{ghirardato2004differentiating}; red area in Figure \ref{fig:MBA_changing}), which includes Choquet expected utility (\citeauthor{schmeidler1989subjective}, \citeyear{schmeidler1989subjective}) and maxmin expected utility (\citeauthor{gilboa1989maxmin}, \citeyear{gilboa1989maxmin}), satisfies certainty independence that is equivalent to constant absolute and relative ambiguity aversion. Similarly, variational preferences (\citeauthor{Variational}, \citeyear{Variational}; green area) satisfy weak certainty independence implying \nameref{ax_CAAA}, while confidence preferences (\citeauthor{chateauneuf2009ambiguity}, \citeyear{chateauneuf2009ambiguity}; purple area) satisfy worst independence that implies \nameref{ax_CRAA}. 

We depart from standard models of ambiguity aversion along two dimensions. First, following \cite{BaillonTestincCAAARA}, we weaken the independence axioms previously mentioned to capture \nameref{ax_DAAA} and \nameref{ax_DRAA}\footnote{We borrow the axiom of \nameref{ax_DAAA} from \cite{Xuechanging}. Ghirardato and Siniscalchi proposed a similar axiom in their work \virgo{Compensanted Absolute Ambiguity Attitudes,} a version of which was presented at RUD and D-TEA conferences in 2015.} (gray area in Figure \ref{fig:MBA_changing}). According to the former, ambiguity becomes more tolerable when the individual is better off in absolute terms, while to the latter, ambiguity becomes more tolerable when the relative size of the ambiguity the individual faces increases.\footnote{As a byproduct, our analysis provides representation results for all the kinds of changing ambiguity attitudes, including increasing absolute and increasing relative ambiguity aversion.} Second, our analysis does not rely on any form of convexity of preferences. This allows us to avoid criticisms related to the Machina paradoxes (\citeauthor{machina2009risk},  \citeyear{machina2009risk}; \citeauthor{baillon2011ambiguity}, \citeyear{baillon2011ambiguity}).


\begin{figure}[ht!]
    \centering
    \includegraphics[scale=0.45]{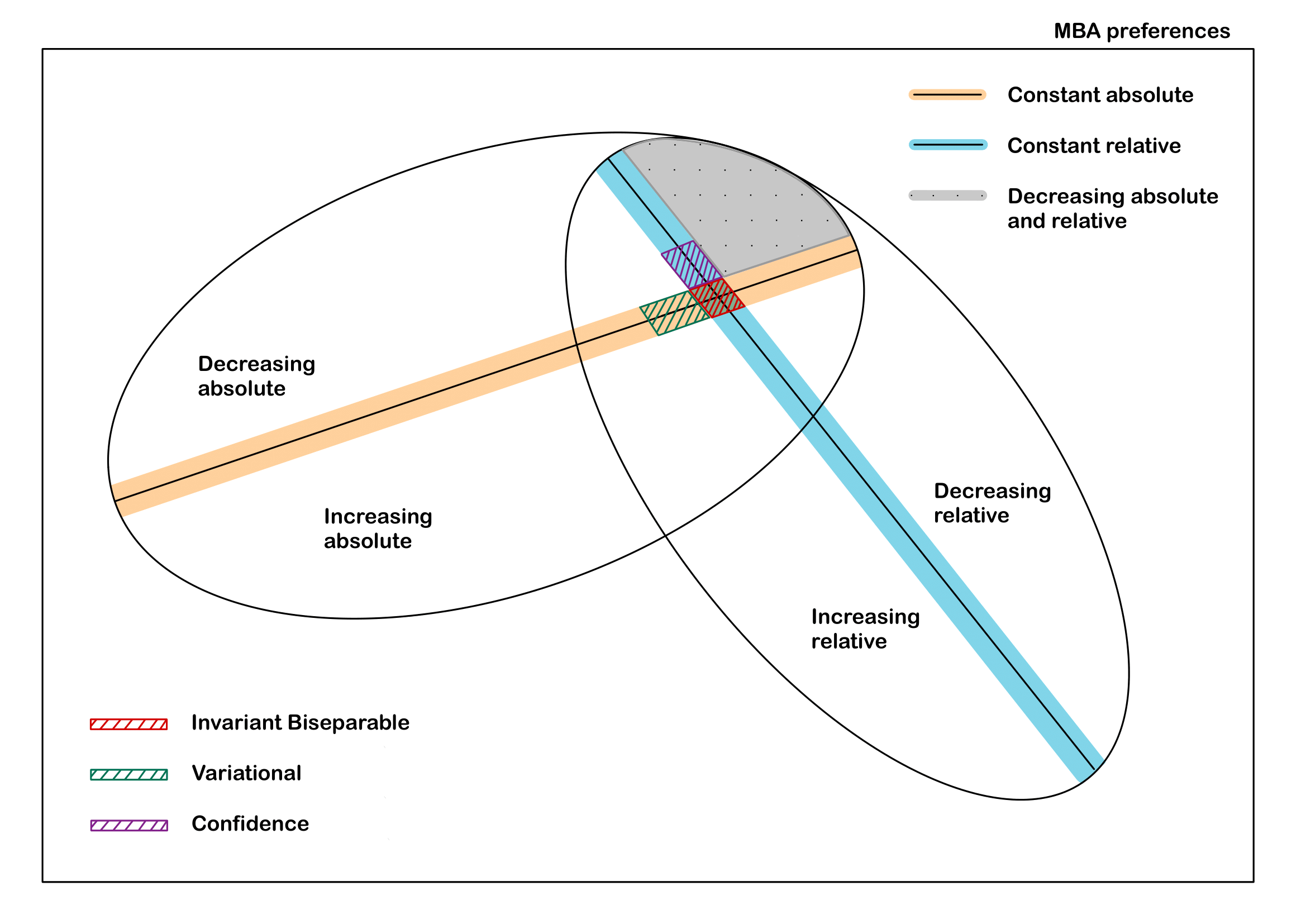}
    \caption{MBA preferences and changing ambiguity attitudes}
    \label{fig:MBA_changing}
\end{figure}


\paragraph{MBA preferences.} Within the standard \citeauthor{AA63} (\citeyear{AA63}) setup, we assume MBA preferences as primitives (\citeauthor{OmnibusSimone}, \citeyear{OmnibusSimone}). On top of completeness and transitivity, these preferences satisfy minimal requirements of rationality such as monotonicity, risk independence, and continuity conditions.\footnote{The acronym \textit{MBA} refers to Monotone, Bernoulli, and Archimedean.}  Notably, we do not employ any unboundedness condition for our results.


As a baseline for our analysis, we characterize MBA preferences (Lemma \ref{Lemma_1}) in terms of certainty equivalent functionals representable as maxima of quasiconcave functions.\footnote{This result generalizes Theorem 4 in \cite{Dualself2022}.} Denote by $\mathcal{F}$ the set of all acts $f: S \to X$ that map the set of states of the world $S$ to set of consequences $X$. Lemma \ref{Lemma_1} shows that the certainty equivalent $I$ associated with act $f \in \mathcal{F}$ can be written as \begin{equation} \label{representation_MBA_intro}
I(u(f)) = \max\limits_{G \in \mathcal{G}} \inf\limits_{p \in \bigtriangleup(S)} G\left(\int u(f) \textnormal{d} p , p \right)
\end{equation} where $u$ is an affine utility function, $\mathcal{G}$ is a set of quasiconvex functions and monotone in the first argument, and $p \in \triangle(S)$ is a probability measure over the state space $S$.

We highlight two feasible interpretations for the representation \eqref{representation_MBA_intro}. The first interpretation, advanced by \cite{Dualself2022} and \cite{xia2020decision}, rationalizes the decision problem as an intrapersonal game between two conflicting \virgo{selves}: Optimism, playing the best possible aggregator in $\mathcal{G}$, and Pessimism, selecting the worst possible belief in $\bigtriangleup(S)$. Alternatively, following \cite{StarShaped_Castagnoli}, a second interpretation relates to the axiomatic treatment of risk mitigation by \cite{dreze1990essays}. In this setting, decision makers are confident that choosing an action induces a probability measure over the state space without error. However, if we remove such confidence, we can relate the decision maker's action as inducing an aggregator in $\mathcal{G}$ and concerns for misspecification to the infimum among all possible probabilistic models in $\bigtriangleup(S)$.



\paragraph{Decreasing absolute ambiguity aversion.} We enrich our baseline (Lemma \ref{Lemma_1}) by considering preferences satisfying \nameref{ax_DAAA}. This axiom states that, if a mixture involving an act $f \in \mathcal{F}$ and a constant act $x \in X$, written $\alpha f + (1-\alpha) x$ for $\alpha \in (0,1)$, is preferred to $\alpha z + (1-\alpha)x$, where $z \in X$, then this preference is preserved if we replace $x$ in both acts with a better constant act $y \in X$.\footnote{As customary, we identify constant acts with their consequence in $X$.}

Theorem \ref{thm:UAPstyleDAAA}, which constitutes our main representation of \nameref{ax_DAAA}, formalizes the link between this property and constant superadditivity. Relating to the representation \eqref{representation_MBA_intro}, Theorem \ref{thm:UAPstyleDAAA} further imposes that each aggregator $ G \in \mathcal{G}$  is constant superadditive in the first argument, that is, $$G\left( \int u(f) + u(x) \textnormal{d} p , p\right) \geq G\left(\int u(f) \textnormal{d} p , p\right) +u(x)$$ for every $f \in \mathcal{F}$, $x \in X$ with $u(x)\geq 0$, and $p \in \bigtriangleup(S)$. This result conforms to our intuition. Under \nameref{ax_DAAA}, $u(f) + u(x)$ is evaluated with a lower degree of ambiguity aversion than $u(f)$. Therefore, combining an uncertain prospect $f$ with a positive, certain one, $x$, yields a higher utility than when the two prospects are considered separately, resulting in constant superadditivity.

Since \nameref{ax_DAAA} weakens weak certainty independence (see section \ref{discussion_independence}), in Proposition \ref{dependent_variational_main} we provide a second characterization of this property which follows an act-dependent variational representation (\citeauthor{Variational}, \citeyear{Variational}). In particular, each act $f \in \mathcal{F}$ is evaluated as 
$$I(u(f))= \max\limits_{c\in C_{f}}\min\limits_{p\in \bigtriangleup(S)}\left\lbrace \int u(f)\textnormal{d}p+c(p)\right\rbrace$$ where the set $C_f$ depends on the act $f$ and collects lower semicontinuous and convex ambiguity costs, $c: \bigtriangleup(S) \to (-\infty, \infty]$. Decreasing absolute ambiguity aversion is captured by the fact that the set of the ambiguity costs enlarges as the decision maker evaluates more favorable acts, i.e., $u(g) \leq u(f)$ implies $C_g \subseteq C_{f}$, for $f, g \in \mathcal{F}$. In a dual-self perspective, as the prospect considered improves, the optimistic self is better off as she can choose an ambiguity cost from a larger set; in the risk mitigation interpretation, the
cost related to misspecification concerns is attenuated.

We introduce two models to illustrate our representations. Example \ref{example1DAAA} adds risk mitigation concerns to second-order expected utility (\citeauthor{grant2009second}, \citeyear{grant2009second}; \citeauthor{neilson2010simplified}, \citeyear{neilson2010simplified}) and shows it satisfies the representation of Theorem \ref{thm:UAPstyleDAAA}; Example \ref{multiplier_benchmark} relates to the representation of Proposition \ref{dependent_variational_main} by introducing an act-dependent version of multiplier preferences (\citeauthor{hansen2001robust}, \citeyear{hansen2001robust}) where the decision maker has access to an outside option. 
\paragraph{Decreasing relative ambiguity aversion.} This notion captures the idea that ambiguity aversion decreases as the proportion of the certainty part of the act decreases. To define it, we assume the existence of a worst consequence $x_* \in X$, i.e., $u(x_*) \leq u(y)$ for all $y \in X$, which we normalize to $u(x_*)=0$. This axiom states that if, for all acts $f \in \mathcal{F},$ $y \in X$, and mixing weight $\alpha \in (0,1)$,  $\alpha f + (1-\alpha) x_*$ is preferred to $\alpha y + (1-\alpha)x_*$, then this preference is preserved if we replace $\alpha$ in both acts with a larger mixing weight $\beta \geq \alpha$.

Our axiom of \nameref{ax_DRAA} is novel; it modifies the one proposed by \cite{Xuechanging}  by restricting to mixtures involving the worst consequence only. To check the soundness of this axiom, we show that, in the context of the smooth ambiguity model (\citeauthor{klibanoff2005smooth}, \citeyear{klibanoff2005smooth}), and assuming ambiguity aversion, it is equivalent to requiring that the function governing ambiguity attitudes satisfies decreasing relative risk aversion (DRRA), i.e., the Arrow-Pratt coefficient of relative risk aversion is decreasing.\footnote{See Remark \ref{concave_DRRA_smooth} for further details. The analog result for \nameref{ax_DAAA} in the context of the smooth model is shown by \cite{Xuechanging} in Proposition 7.} 

Theorem \ref{representation_IRAA}, our main representation of \nameref{ax_DRAA}, connects this property to positive superhomogeneity. Relating to the representation \eqref{representation_MBA_intro}, Theorem \ref{representation_IRAA} imposes that each aggregator $ G \in \mathcal{G}$ is positively superhomogeneous in the first argument, namely, $$G\left( \int u(\alpha f+(1-\alpha)x_*) \textnormal{d} p , p\right) \leq \alpha \> G\left(\int u(f) \textnormal{d} p , p\right)$$ for every $\alpha \in (0, 1)$, $f \in \mathcal{F}$, and $p \in \bigtriangleup(S)$. When considering $\alpha f+(1-\alpha)x_*$, as $\alpha$ increases, the proportion of the certainty part of the act decreases. Therefore, under \nameref{ax_DRAA}, $\alpha f+(1-\alpha)x_*$ is evaluated with a higher degree of ambiguity aversion than $f$, resulting in positive superhomogeneity. 



Since \nameref{ax_DRAA} is a weakening of worst independence (see section \ref{discussion_independence}), Proposition \ref{dependent_confidence_main} provides an alternative characterization of this property in the form of an act-dependent confidence preference representation (\citeauthor{chateauneuf2009ambiguity}, \citeyear{chateauneuf2009ambiguity}). In particular, each act $f \in \mathcal{F}$ is evaluated as follows
$$I(u(f))=\max\limits_{d\in D_{f}}\min\limits_{p\in \bigtriangleup(S)}\frac{\int_{S}u(f)\textnormal{d}p}{d(p)}$$ where the set $D_f$ depends on the act $f$ and collects upper semicontinuous and quasiconcave confidence functions, $d: \bigtriangleup(S) \to [0, \infty]$. As for absolute attitudes, \nameref{ax_DRAA} is captured by the set of confidence functions enlarging as the decision maker evaluates more favorable acts, i.e., $u(g) \leq u(f)$ implies $D_g \subseteq D_{f}$, for $f, g \in \mathcal{F}$. This feature implies positive superhomogeneity conforming to Theorem \ref{representation_IRAA}.

Two models illustrate our representations of \nameref{ax_DRAA}. Analogously to Example \ref{example1DAAA}, Example \ref{example1DRAA} connects second-order expected utility with risk mitigation to Theorem \ref{representation_IRAA}; Example \ref{confidence_benchmark} relates to Proposition \ref{dependent_confidence_main} by introducing an act-dependent model of entropic-confidence preferences with outside option.







\paragraph{Risk sharing application.} We apply our characterizations of absolute and relative ambiguity attitudes to investigate whether, in financial markets, it is efficient for agents displaying general ambiguity preferences to engage in speculative betting. In an exchange economy with a single consumption good and no aggregate uncertainty, subjective expected utility agents that are risk averse introduce individual uncertainty into the final allocation--- meaning, they engage in betting--- if and only if their beliefs differ (\citeauthor{milgrom1982information}, \citeyear{milgrom1982information}). This result has been extended, first by \cite{billot2000sharing} in the context of maxmin expect utility, and later by \cite{rigotti2008subjective} for general convex preferences, by showing that agents will bet if and only if they do not share \textit{any} beliefs, i.e., their sets of subjective beliefs, properly defined, do not intersect. Finally, \cite{ghirardato2018risk} generalize this result to non-convex preferences, proposing a condition called \textit{strict pseudoconcavity at certainty} which, loosely speaking, requires the indifference curves at a consumption point to lie strictly above the tangent line. 

Proposition \ref{risksharingth} connects these risk sharing results to our analysis of changing ambiguity attitudes. In particular, it shows that by combining our characterizations of \nameref{ax_DAAA} and increasing relative ambiguity aversion, plus some additional regularity requirements, we obtain strict pseudoconcavity at certainty.\footnote{The representation of increasing relative ambiguity aversion follows that of Theorem \ref{representation_IRAA} for \nameref{ax_DRAA} by replacing positively superhomogeneity with positively subhomogeneity.} The inefficiency of betting is then implied by \cite{ghirardato2018risk}.


In contrast, in Example \ref{DRAA_betting}, we argue that \nameref{ax_DRAA} might prompt agents to bet even if they share the same beliefs. The key intuition is that Theorem \ref{representation_IRAA} links \nameref{ax_DRAA} to a representation that is ``convex at $0$,'' leading agents to prefer uncertainty over full insurance.

\paragraph{Related literature.} We contribute directly to the decision-theoretic literature investigating changing ambiguity attitudes resulting from utility shifts. \cite{grant2013mean} focus on the case of \nameref{ax_CAAA} and show that this property admits a mean dispersion representation. We instead consider non-constant absolute and relative attitudes. To this end, we borrow the axioms of \nameref{ax_DAAA} and \nameref{ax_DRAA}, the latter with some modifications, from the analysis of \cite{Xuechanging}.\footnote{See also \cite{chambers2014two} for the related notion of absolute uncertainty attitudes.} Like us, \cite{Xuechanging} studies changing ambiguity attitudes resulting from utility shifts, but focuses only on the case of ambiguity averse preferences. 

A different approach is pioneered by \cite{Cherbonnier2015DecreasingAU} and \cite{CMMRwealth}, which consider wealth effects, instead of utility shifts. In this framework, \cite{CMMRwealth} show that wealth-classifiable preferences, that is, ambiguity preferences that are either increasing, constant, or decreasing to wealth changes in absolute or relative terms, must necessarily display constant absolute risk aversion. Therefore, one can view studies on utility shifts as a way to overcome this restrictive assumption on risk preferences. Employing a different notion of comparative ambiguity attitudes, \cite{wang2019comparative} studies wealth effects without assumptions on risk preferences. 

Our work is motivated by the recent experimental evidence of \cite{BaillonTestincCAAARA}. At the aggregate level, their findings support decreasing absolute and decreasing relative ambiguity
aversion. Roughly 40$\%$ of the subjects of their experiments display \nameref{ax_CAAA} and  another 40$\%$ display \nameref{ax_DAAA}; almost half
of the subjects satisfy decreasing relative ambiguity aversion. Furthermore, their analysis suggests that \nameref{ax_CAAA} would not make accurate
predictions for most subjects unless we accept errors of up to 10$\%$. This evidence highlights the inadequacy of popular decision-making models
to capture experimentally corroborated ambiguity attitudes. See section \ref{discussion_independence} for a discussion. 


A growing literature studies the consequences of \nameref{ax_DAAA} in economic applications. Some of these applications interpret this property as a form of ambiguity prudence and use it to explain precautionary savings (\citeauthor{berger2014precautionary}, \citeyear{berger2014precautionary}; \citeauthor{osaki2014precautionary}, \citeyear{osaki2014precautionary}) and self-protective behavior (\citeauthor{berger2016impact}, \citeyear{berger2016impact}). Other applications study market selection. \cite{guerdjikova2015survival} show that decision makers displaying \nameref{ax_DAAA} survive in markets populated by expected utility agents. In a framework that allows for a broad class of recursive preferences,  \cite{beker2023if} show that under \nameref{ax_DAAA}, every full-support belief can survive if there is sufficiently high uncertainty. As for relative attitudes, in the context of uncertainty sharing economies, \cite{hara2024sharing} show that if the individual consumers display \nameref{ax_CRAA}, then the representative consumer exhibits \nameref{ax_DRAA}.


We borrow technical insights from the literature on risk measures. Some of our results are inspired by \cite{RuoduCashSubad} that provide representations as minima of constant subadditive and quasiconvex functions. Similar results for positively superhomogeneous functionals appear in: \cite{StarShaped_Castagnoli}, which characterize star-shaped monetary, i.e., constant additive, risk measures;\footnote{\cite{nonlinear_SimoneGiacomoBoberto} retrieve analogous results with additional subdifferential properties.} \cite{laeven2023dynamic}, which provide a representation of star-shaped functionals in terms of convex and lower semicontinuous functionals. Our results distinguish from these since, without assuming unboundedness, we have to employ extension techniques and envelope continuity results.


\section{Mathematical preliminaries}

We consider the \cite{AA63} setup composed of a nonempty set $S$ of \textit{states of the world}, endowed with an algebra $\Sigma$ of subsets of $S$ called \textit{events}, and a nonempty convex set $X$ of \textit{consequences}. Denote by $\bigtriangleup(S)$ the set of finitely additive probability measures over $S$.

The decision maker has preferences over the set $\mathcal{F}$ of all \textit{(simple) acts}, i.e., $\Sigma$-measurable functions $f:S\to X$ such that $f(S)$ is a finite set. For all $x\in X$, we identify $x\in \mathcal{F}$ as the constant act equal to $x$, and, as a result, $X$ as a subset of $\mathcal{F}$. For all $f,g\in\mathcal{F}$, and $\alpha\in [0,1]$, relying on the linear structure of $X$, we define convex combinations of acts as \[
(\alpha f+(1-\alpha)g):s\mapsto \alpha f(s)+(1-\alpha)g(s)\in X.
\]

We denote by $\succsim$ a binary relation over $\mathcal{F}$, and by $\succ$ and $\sim$ its asymmetric and symmetric parts, respectively. For all $f\in \mathcal{F}$, denote by $x_f\in X$ a \textit{certainty equivalent} of $f$, i.e., $x_f\sim f.$ A function $V: \mathcal{F} \to \R$ is a \textit{utility representation} for $\succsim$ if, for all $f, g \in \mathcal{F}$, $$f \succsim g \iff V(f) \geq V(g).$$  Whenever the decision maker's preferences over consequences admit a utility representation, comparisons between acts can be expressed in utility levels. To formalize this, we introduce $B_0(\Sigma, K)$, the set of $\Sigma$-measurable real-valued bounded simple functions whose images are included in $K$, interpreted as the set of utility levels, for some $K\subseteq \mathbb{R}$. To ease notation, let $B_0(\Sigma)=B_0(\Sigma,\mathbb{R})$. Endow these sets with the supnorm $\lVert\cdot\rVert_{\infty}$. 

As customary in decision theory under ambiguity, we study the properties of ``certainty equivalent functionals.'' To this end, we introduce some properties discussed in the upcoming sections. Fix a convex set $K\subseteq \R$. A functional $I:\BK\to \bar{\R}$ is \textit{normalized} if $I(k)=k$ for all $k\in K$; \textit{monotone} if $\varphi\geq \psi$ implies $I(\varphi)\geq I(\psi)$, for all $\varphi,\psi\in \BK$, and \textit{strictly} monotone if  in addition $\varphi\neq\psi$ implies $I(\varphi)>I(\psi)$; \textit{quasiconcave} if for all $\varphi,\psi\in \BK$, and $\alpha\in [0,1]$, we have
\[
I(\alpha \varphi+(1-\alpha)\psi)\geq \min\left\lbrace I(\varphi),I(\psi)\right\rbrace;
\]
\textit{quasiconvex} if $-I:B_0(\Sigma,K) \to \bar{\R}$ is quasiconcave. 

The notions of constant super and subadditivity are particularly relevant for our analysis of absolute ambiguity attitudes. Specifically, $I$ is \textit{constant superadditive} if, for all $\varphi\in \BK$, and $k\geq 0$, such that $\varphi+k\in \BK$, we have 
\begin{equation}\label{const_sup}
I(\varphi+k)\geq I(\varphi)+k.
\end{equation}

\noindent
Analogously, $I$ is \textit{constant subadditive} if equation \eqref{const_sup} holds with the reverted inequality; \textit{constant additive} if it is both constant superadditive and constant subadditive. Constant superadditivity and subadditivity are \textit{strict} when they hold with the strict inequality for every $\varphi\neq 0$ and $k>0$.

For relative ambiguity attitudes, we introduce the concepts of positive super and subhomogeneity.
In particular, $I$ is \textit{positively superhomogeneous} if
\begin{equation} \label{super_homo}
  I(\gamma \varphi)\leq \gamma I(\varphi)  
\end{equation}
for all $\gamma\in (0,1)$, and $\varphi\in \BK$ with $\gamma \varphi\in \BK$; \textit{positive subhomogeneity} is defined by reverting the inequality \eqref{super_homo}; \textit{positive homogeneity} holds if $I$ is both positively super and subhomogeneous. Positive superhomogeneity and subhomogeneity    are \textit{strict} when they hold with the strict inequality for every $\varphi\neq 0$.

Our representations satisfy additional regularity requirements. A family $\Upsilon$ of functions $H:B_0(\Sigma,K) \to \bar{\R}$ is \textit{regular} if: \textit{(i)} $\varphi \mapsto \max_{H\in \Upsilon}H(\varphi)$ is well-defined and continuous, \textit{(ii)} if $K$ is lower open,\footnote{A subset $K$ of $\R$ is \textit{lower open} (resp. \textit{upper open}) if, for all $k \in K$, there exists $\varepsilon>0$ such that $[k-\varepsilon,k]\subseteq K$ (resp. $[k,k+\varepsilon]\subseteq K$).} then each element of $\Upsilon$ is lower semicontinuous, and \textit{(iii)} if $K$ is either upper open or a bounded interval, then each element of $\Upsilon$ is upper semicontinuous. Moreover, a family $\mathcal{G}$ of functions $G:\R\times \bigtriangleup(S)\to \bar{\R}$ is \textit{linearly continuous} if 
\[
\varphi\mapsto \max\limits_{G\in \mathcal{G}}\inf\limits_{p\in \bigtriangleup(S)}G\left(\int \varphi \textnormal{d}p,p\right)
\]
is continuous.

\section{Monotone, Bernoulli, Archimedean preferences} \label{MBA_preferences_section}


We now introduce the axioms, maintained throughout our analysis, characterizing MBA preferences. Let a binary relation $\succsim$ on $\mathcal{F}$ represent the decision maker's preferences.

\begin{axiom}[weak order]\label{ax_weak_order}
$\succsim$ is complete and transitive.
\end{axiom}

\begin{axiom}[risk independence]\label{ax_risk_ind}
If $x,y,z\in X$, and $\alpha\in (0,1)$,
\[
x\sim y\implies \alpha x+(1-\alpha)z\sim \alpha y+(1-\alpha)z.
\]
\end{axiom}

\begin{axiom}[archimedean continuity]\label{ax_arch_cont}
If $f,g,h\in \mathcal{F}$ and $f\succ g\succ h$, then there exist $\alpha,\beta\in [0,1]$ such that
\[
\alpha f+(1-\alpha)h\succ g\succ \beta f+(1-\beta)h.
\]
\end{axiom}

\begin{axiom}[monotonicity]\label{ax_monot}
If $f,g\in \mathcal{F}$ and $f(s)\succsim g(s)$ for each $s\in S$, then $f\succsim g$.
\end{axiom}

A binary relation $\succsim$ on $\F$ is an \textit{MBA preference}, or simply \textit{MBA}, if it satisfies \nameref{ax_weak_order}, \nameref{ax_risk_ind}, \nameref{ax_arch_cont}, and \nameref{ax_monot}.



Lemma \ref{Lemma_1} below provides a representation of MBA preferences which constitutes the starting point of our analysis and connects existing results in the literature.
\cite{OmnibusSimone} show that a binary relation $\succsim$ is MBA if and only if there exist: \textit{(i)} an affine function $u:X\to \mathbb{R}$, and \textit{(ii)} a monotone, normalized, and continuous functional $I:B_0(\Sigma,u(X))\to \mathbb{R}$ such that $I\circ u$ represents $\succsim$. We sharpen their representation showing that the certainty equivalent $I$ can be taken as the maximum over a set of monotone and quasiconcave functions. A similar result appears also in \cite{Dualself2022} (Theorem 4), but limited to finite state spaces. We extend their result to arbitrary state spaces and establish additional regularity conditions.

\begin{lemma} \label{Lemma_1}
Let $\succsim$ be a binary relation over $\mathcal{F}$. The following are equivalent
\begin{enumerate}[label=(\roman*)]
    \item $\succsim$ is an MBA preference relation.
      \item There exist an affine function $u:X\to \mathbb{R}$, and a regular set $\Psi$ of monotone, quasiconcave functions $H:B_0(\Sigma,u(X))\to \R$ such that
    \begin{equation} \label{MBA_pref}
      f\succsim g\Longleftrightarrow \max\limits_{H\in \Psi}H(u(f))\geq \max\limits_{H\in \Psi}H(u(g))  
    \end{equation}
    for all $f,g\in \F$ and $\max_{H\in \Psi}H(u(x))=u(x)$ for all $x\in X$. 
\end{enumerate} 
\end{lemma}

We highlight two feasible interpretations for our representation of MBA preferences. Using quasiconcave duality results by \cite{Comp_mon_quasic}, the utility representation in \eqref{MBA_pref} can be written as \begin{equation}
    \label{Lemma1_G} V(f) = \max\limits_{G \in \mathcal{G}} \inf\limits_{p \in \bigtriangleup(S)} G\left(\int u(f) \textnormal{d} p , p \right)
\end{equation}
where $\mathcal{G}$ is a linearly continuous family of quasiconvex functions monotone in the first argument. The first interpretation, by \cite{Dualself2022} and \cite{xia2020decision}, rationalizes the decision problem as an intrapersonal game between two conflicting \virgo{selves}: Optimism, playing the best possible aggregator in $\mathcal{G}$, and Pessimism, selecting the worst possible belief in $\bigtriangleup(S)$. Alternatively, following \cite{StarShaped_Castagnoli}, a second interpretation relates to the axiomatic treatment of risk mitigation by \cite{dreze1990essays}. In Dr\`eze, decision makers are confident that choosing an action induces a probability measure over the state space without error. By allowing for less confident decision makers, we interpret the decision maker's action as inducing an aggregator in $\mathcal{G}$ and concerns for misspecification as the infimum among all possible probabilistic models in $\bigtriangleup(S)$.

In the following sections, we enrich our representation of MBA preferences by studying the role of absolute and relative ambiguity attitudes. To this end, we first discuss how other ambiguity models relate to these properties.    

\subsection{Independence notions and ambiguity attitudes} \label{discussion_independence}

Popular models based on MBA preferences are neutral towards absolute and/or relative changes in utility levels. This neutrality is implied by distinct notions of independence, stronger than \nameref{ax_risk_ind}. To discuss this aspect further, we formally introduce \nameref{ax_CAAA} (\citeauthor{grant2013mean}, \citeyear{grant2013mean}), and \nameref{ax_CRAA}. 

\begin{axiom}[constant absolute ambiguity aversion]\label{ax_CAAA}
For all $f\in\mathcal{F}$, $x,y,z\in X$, and $\alpha\in (0,1)$,  \[
\alpha f+(1-\alpha)x\succsim \alpha z+(1-\alpha)x\implies \alpha f+(1-\alpha)y\succsim \alpha z+(1-\alpha)y.
\]
\end{axiom}

To define relative ambiguity attitudes, we assume that the binary relation $\succsim$ on $\mathcal{F}$ admits a \textit{worst consequence}. That is, there exists $x_*\in X$ such that $y\succsim x_*$ for all $y\in X$.

\begin{axiom}[constant relative ambiguity aversion]\label{ax_CRAA}
For all $f\in\mathcal{F}$, $x,y\in X$, and $\alpha,\beta\in (0,1)$, 
\[
\alpha f+(1-\alpha)x_*\succsim \alpha y+(1-\alpha)x_*\implies \beta f+(1-\beta)x_*\succsim \beta y+(1-\beta)x_*.
\]
\end{axiom}

Both these axioms impose a form of \virgo{independence.} Constant absolute ambiguity aversion implies independence in absolute changes in utility levels, meaning the decision maker's ambiguity aversion is invariant towards absolute utility shifts. To see this, notice that the implication in the axiom does not depend on the preference ranking of the constant acts $x$ and $y$. Similarly, constant relative ambiguity aversion requires independence towards relative changes in utility levels. This follows since the implication in the axiom does not depend on whether the weight $\alpha$ is larger or smaller than the weight $\beta$.

Below, we summarize the independence notions of popular MBA preference models, highlighting their relation to constant absolute and constant relative ambiguity attitudes.  

\begin{itemize}[leftmargin=0.4cm, itemsep=0.3pt, topsep=0.5pt]

\item[-] \textit{Choquet (\citeauthor{schmeidler1989subjective}, \citeyear{schmeidler1989subjective}), maxmin  (\citeauthor{gilboa1989maxmin}, \citeyear{gilboa1989maxmin})}, \textit{$\alpha$-maxmin (\citeauthor{ghirardato2004differentiating}, \citeyear{ghirardato2004differentiating}), and dual-self maxmin (\citeauthor{Dualself2022}, \citeyear{Dualself2022}).} 

All these models imply \textit{certainty independence}, which requires that preferences are independent of mixtures with constant acts. Formally, a binary relation $\succsim$ on $\mathcal{F}$ satisfies certainty independence if, for all $f, g \in \mathcal{F}$, $x \in X$, and $\alpha \in [0,1]$, $$f \succsim g \iff \alpha f + (1-\alpha)x \succsim \alpha g + (1-\alpha)x.$$ It can be seen how certainty independence implies both \nameref{ax_CAAA} and \nameref{ax_CRAA}. 

\item[-] \textit{Variational (\citeauthor{Variational}, \citeyear{Variational}), and  vector expected utility (\citeauthor{siniscalchi2009vector}, \citeyear{siniscalchi2009vector})}.

These preferences satisfy a weaker version of certainty independence, named \textit{weak certainty independence}. A binary relation $\succsim$ on $\mathcal{F}$ satisfies weak certainty independence if, for all $f, g \in\mathcal{F}$, $x,y\in X$, and $\alpha\in (0,1)$,  \[
\alpha f+(1-\alpha)x\succsim \alpha g+(1-\alpha)x\implies \alpha f+(1-\alpha)y\succsim \alpha g+(1-\alpha)y.
\] It is immediate to verify that this notion implies \nameref{ax_CAAA}. Moreover, it allows for dependence of ambiguity aversion on relative utility shifts. Indeed, this notion requires that preferences are independent of  mixtures with constant acts while keeping the relative mixing weights, $\alpha$ and $1- \alpha$, constant.\footnote{Variational preferences are represented by a concave and normalized certainty equivalent and hence satisfy increasing relative ambiguity aversion. This contrasts with the maxmin model, where the certainty independence axiom implies positive homogeneity and, in turn, constant relative ambiguity aversion.}


\item[-] \textit{Confidence preferences (\citeauthor{chateauneuf2009ambiguity}, \citeyear{chateauneuf2009ambiguity}).} 

This model satisfies \textit{worst independence}, imposing independence to mixtures involving the worst consequence. Formally, a binary relation $\succsim$ on $\mathcal{F}$ satisfies worst independence if, for all $f, g \in \mathcal{F}$, and $\alpha \in (0, 1)$, $$f \sim g \implies \alpha f + (1 - \alpha) x^* \sim \alpha g + (1 - \alpha) x^*.$$
For MBA preferences, worst independence is equivalent to \nameref{ax_CRAA}, as both lead to positively homogeneous certainty equivalents. Furthermore, specularly to weak certainty independence, this axiom allows for the dependence of ambiguity aversion on absolute utility shifts by requiring independence only with respect to relative mixing weights, keeping fixed the worst consequence on all mixtures.\footnote{Since such preferences are represented by a superlinear and normalized functional, they exhibit decreasing absolute ambiguity aversion.}
\end{itemize}


\medskip
\noindent
In our analysis, we weaken the axioms of weak certainty independence and worst independence to allow for changing absolute and relative ambiguity attitudes, respectively. In doing so, we provide representations that are more general than the ones implied by these axioms: Proposition \ref{dependent_variational_main} connects \nameref{ax_DAAA} to an act-dependent variational representation, while Proposition \ref{dependent_confidence_main} links \nameref{ax_DRAA} to an act-dependent confidence preference representation.\footnote{By weakening certainty independence, \cite{hill2013confidence} provides an axiomatization of an act-dependent version of maxmin expected utility to capture the role of confidence in decisions.}

\section{Absolute ambiguity attitudes}

In this section, we characterize preferences that display less aversion to ambiguity as the decision maker's baseline utility increases. To this end, we employ \citeauthor{Xuechanging}'s (\citeyear{Xuechanging}) axiom of decreasing absolute ambiguity aversion.

\begin{axiom}[decreasing absolute ambiguity aversion]\label{ax_DAAA}
For all $f\in\mathcal{F}$, $x,y,z\in X$, and $\alpha\in (0,1)$, if $y\succsim x$, then
\[
\alpha f+(1-\alpha)x\succsim \alpha z+(1-\alpha)x\implies \alpha f+(1-\alpha)y\succsim \alpha z+(1-\alpha)y.
\] \end{axiom}

Contrary to the constant absolute case, \nameref{ax_DAAA} allows the decision maker's ambiguity aversion to depend on absolute changes in utility levels. In particular, the axiom says that if an ambiguous act $\alpha f+(1-\alpha)x$ is preferred over a constant act $\alpha z+(1-\alpha)x$, then such a ranking is preserved if the certainty part improves from $x$ to $y$ on both sides.   This axiom captures the idea that ambiguity becomes more
tolerable when the decision maker is better off in absolute terms. 

The following theorem provides our main representation of \nameref{ax_DAAA}. It connects this property to the representation \eqref{Lemma1_G} of Lemma \ref{Lemma_1} by imposing constant superadditivity of all aggregators.


\begin{theorem}\label{thm:UAPstyleDAAA}
Let $\succsim$ be a binary relation over $\mathcal{F}$. The following are equivalent
\begin{enumerate}[label=(\roman*)]
    \item $\succsim$ is MBA and exhibits \nameref{ax_DAAA}.
    \item There exist an affine function $u:X\to \mathbb{R}$ and a linearly continuous family $\mathcal{G}$ of monotone, constant superadditive in the first argument, and quasiconvex functions $G:\mathbb{R}\times \bigtriangleup(S)\to \bar{\R}$ such that
    \[
    f\succsim g \Longleftrightarrow  \max_{G\in \mathcal{G}}\inf_{p\in \bigtriangleup(S)}G\left(\int u(f)\textnormal{d}p,p\right)\geq \max_{G\in \mathcal{G}}\inf_{p\in \bigtriangleup(S)}G\left(\int u(g)\textnormal{d}p,p\right)
    \]
    for all $f,g\in \mathcal{F}$ and $\max_{G\in \mathcal{G}}\inf_{p\in \bigtriangleup(S)}G\left(\int u(x)\textnormal{d}p,p\right)=u(x)$ for all $x \in X$.
\end{enumerate}
\end{theorem}

The two interpretations provided for Lemma \ref{Lemma_1} apply also for Theorem \ref{thm:UAPstyleDAAA}. The main distinction is the fact that now the representation captures \nameref{ax_DAAA} through the constant superadditivity of each aggregator $G(\cdot, p)$ for all $p \in \bigtriangleup(S)$. Interpreting again the decision problem as an intrapersonal game between two selves (\citeauthor{Dualself2022}, \citeyear{Dualself2022}), this means that the pessimistic self exhibits \nameref{ax_DAAA} irrespective of the move of the optimistic self, that is, the choice of the aggregator. Instead, in a risk mitigation framework (\citeauthor{dreze1990essays}, \citeyear{dreze1990essays}), it is as if the decision maker exhibits \nameref{ax_DAAA} independently of the chosen action. Both interpretations have a common ground: \nameref{ax_DAAA} ascribes to the aggregators only.

The following example introduces a generalization of second-order expected utility (\citeauthor{grant2009second}, \citeyear{grant2009second}; \citeauthor{neilson2010simplified}, \citeyear{neilson2010simplified}), which allows for multiple probability distributions rather than a single one. We show that whenever the function governing ambiguity attitudes satisfies decreasing absolute risk aversion (DARA), i.e., the Arrow-Pratt coefficient of absolute risk aversion is decreasing, then the model can be written as the representation of Theorem \ref{thm:UAPstyleDAAA}, hence displaying \nameref{ax_DAAA}.

\begin{example}[second-order expected utility with risk mitigation; DARA]\label{example1DAAA} We model a decision maker with a finite set of probabilities $Q\subseteq \bigtriangleup^{\sigma}(S)$,\footnote{We denote by $\bigtriangleup^{\sigma}(S)$ the set of countably additive probability measures over $S$.} where $(S,\Sigma)$ is a measurable space, an affine utility over consequences satisfying $u(X)=[0,\infty)$, and ambiguity attitudes are represented by a continuous, strictly increasing, and concave function $\phi:[0,\infty)\to \R$. For every act $f\in \mathcal{F}$, the second-order expected utility given $q\in Q$ is defined as \[
I_q(u(f))=\phi^{-1}\left(\int \phi(u(f))\textnormal{d}q\right).
\] To aggregate models in $Q$, the decision maker uses the following criterion $$V(f)=\max_{q\in Q}I_q(u(f)).$$ Clearly, if $Q$ is a singleton, these preferences collapse to second-order expected utility. In general, inspired by \cite{dreze1990essays}, we interpret these preferences as reflecting a two-step procedure: first, the choice of an ambiguous alternative $f \in \mathcal{F}$, and second the choice of an action $q \in Q$ which partially controls the probability over the states to mitigate the uncertainty involving $f$. For this reason, we refer to these preferences as \textit{second-order expected utility with risk mitigation}. 



Through the properties of our representation of Theorem \ref{thm:UAPstyleDAAA}, we investigate which assumptions on $\phi$ imply \nameref{ax_DAAA}. As each $I_q$ is continuous and $Q$ is finite, $V$ is also continuous; as each $I_q$ is monotone and normalized, $V$ is monotone and normalized. By Lemma \ref{Lemma_constsuppossup_examples} and \ref{representation_for_examples} in Appendix \ref{Appendix_B3}, we have that, if $\phi$ is twice differentiable and satisfies DARA, i.e., $t\mapsto -\phi''(t)/\phi'(t)$ is decreasing, then each $I_q$ is constant superadditive.\footnote{The function $\phi:t\mapsto\sqrt{t}$ satisfies all the listed hypotheses.} 
As a consequence, such preferences exhibit \nameref{ax_DAAA} and admit a representation as in Theorem \ref{thm:UAPstyleDAAA}, with
\[
G_q(t,p)=\sup\left\lbrace I_q(u(f)):f\in \mathcal{F}\ \textnormal{and}\ \int u(f) \textnormal{d}p\leq t \right\rbrace
\]
for all $(t,p)\in \mathbb{R}\times \bigtriangleup(S)$. The constant superadditivity of each $I_q$ implies that each $G_q(\cdot,p)$ is constant superadditive as well. \hfill $ \blacktriangleleft$ 



\end{example}






As \nameref{ax_DAAA} weakens weak certainty independence (see section \ref{discussion_independence}), we investigate whether these preferences admit a more general variational representation. The following proposition shows that this is the case by relating \nameref{ax_DAAA} to an act-dependent variational representation.\footnote{Proposition \ref{dependent_variational_main} extends Proposition A.3 in \cite{RuoduCashSubad} by removing unboundedness conditions.} In their standard formulation (\citeauthor{Variational}, \citeyear{Variational}), variational preferences are characterized by a single ambiguity cost $c:\bigtriangleup(S)\to [0,\infty]$ capturing the level of ambiguity aversion. Due to changing ambiguity attitudes and the absence of convexity of preferences, the set of ambiguity costs varies with the acts.

\begin{proposition} \label{dependent_variational_main} Let $\succsim$ be a binary relation over $\mathcal{F}$. The following are equivalent
\begin{enumerate}[label=(\roman*)]
    \item \label{1prop1} $\succsim$ is MBA and exhibits \nameref{ax_DAAA}.
    \item \label{2prop1} There exist an affine function $u:X\to \mathbb{R}$, and, for all $f\in \mathcal{F}$, a family $C_{f}$ of lower semicontinuous and convex functions $c:\bigtriangleup(S)\to (-\infty,\infty]$ such that
\[f \succsim g \iff 
\max\limits_{c\in C_{f}}\min\limits_{p\in \bigtriangleup(S)}\left\lbrace \int_Su(f)\textnormal{d}p+c(p)\right\rbrace \geq \max\limits_{c\in C_{g}}\min\limits_{p\in \bigtriangleup(S)}\left\lbrace \int_Su(g)\textnormal{d}p+c(p)\right\rbrace,
\] where $\max\limits_{c\in C_x}\min\limits_{p\in\bigtriangleup(S)}c(p)=0$ for all $x\in X$, and $C_{g}\subseteq C_{f}$ for all $f,g\in \mathcal{F}$ with $u(f)\geq u(g)$.
\end{enumerate}\end{proposition} 





This result generalizes the dual-self variational representation in \cite{Dualself2022} (Theorem 3), which differs from ours as it is act-independent, a consequence of \nameref{ax_CAAA}. Instead, due to \nameref{ax_DAAA}, the set of ambiguity costs varies with the act, enlarging as the utility levels increase. In the intrapersonal game interpretation, this maps to the optimistic self being allowed to select a more favorable cost, while in the risk mitigation one, to an attenuation of the costs related to misspecification concerns. In general, the result highlights how the higher the utility levels, the lower the ambiguity costs the decision maker faces. 

In Proposition \ref{dependent_variational_main}
point \textit{\ref{2prop1}}, each cost function belonging to $C_f$ takes values in $(-\infty,\infty]$. This is done in light of the following remark, which provides a nice construction of these costs. Alternatively, each cost function could be taken with values in $[0,\infty]$ and the construction would rely, following \cite{Variational} and \cite{Niveloidsextension}, on the more standard Fenchel-Moreau representation.

\begin{remark} As observed by \cite{RuoduCashSubad} (Proposition A.3), the family of ambiguity costs in the act-dependent variational representation can be written in a specific form. In particular, by inspecting the proof of their Proposition A.3, it follows that
\begin{equation}\label{nonnepossopiuuccidetemi}
I(u(f))=\max\limits_{g:u(f)\geq u(g)}\min\limits_{p\in \bigtriangleup(S)}\left\lbrace \int_Su(f)\textnormal{d}p+u(x_g)-\int_S u(g)\textnormal{d}p\right\rbrace
\end{equation}
for all $f\in \mathcal{F}$.
This representation, which is equivalent to that of Proposition \ref{dependent_variational_main}, has the advantage of making explicit the following structure of ambiguity costs 
\[
C_f=\left\lbrace p\mapsto u(x_g)-\int_S u(g)\textnormal{d}p:g\in \mathcal{F},\ \forall s\in S,\ f(s)\succsim g(s)\right\rbrace.
\]
\[
c^i_g:p\mapsto u(x^i_g)-\int_Su(g)\textnormal{d}p
\]
Notice that, even considering the more explicit form of this representation, in the absence of uniqueness-related results, many other families of ambiguity indexes may exist. \hfill $ \blacktriangleleft$
\end{remark}

The following example relates to the representation of \nameref{ax_DAAA} of Proposition \ref{dependent_variational_main} by considering an act-dependent version of multiplier preferences (\citeauthor{hansen2001robust}, \citeyear{hansen2001robust}). 

\begin{example}[multiplier preferences with outside option] \label{multiplier_benchmark} 
We model a decision maker that can mitigate the uncertainty induced by any act $f \in \mathcal{F}$ by exerting some effort that affects the probability distribution over a finite set of states $S$, where $\Sigma=2^S$. In particular, exerting effort to mitigate the uncertainty of $f$ is valuable if it can induce a model, from the finite set $Q \subseteq \bigtriangleup(S)$, which makes $f$ preferred over an outside option yielding a utility level equal to $\theta \in [0, \infty)$. If no model in $Q$ justifies the choice of $f$ over the outside option, then exerting effort is not valuable, and the decision maker evaluates $f$ using as benchmark $q_{\text{u}}$, the uniform distribution over $S$.
Formally, each act $f$ is associated with a set of probability distributions
\[
C_f=\left\lbrace q\in Q:\int u(f)\textnormal{d} q\geq \theta \right\rbrace\cup \left\lbrace q_{\text{u}}\right\rbrace.
\]
Notice that, for $f, g \in \mathcal{F}$, $u(f) \leq u(g)$ implies $C_f \subseteq C_g$. 

The decision maker is concerned with model misspecification and displays multiplier preferences
\begin{equation} \label{multiplier_pref_oo}
  V(f)=\max\limits_{q\in C_f}\min\limits_{p\in \bigtriangleup(S)}\left\lbrace \int u(f)\textnormal{d} p +\lambda R(p\lVert q) \right\rbrace  
\end{equation}
where $\lambda >0$, and $R(\cdot\lVert\cdot)$ denotes the relative entropy which can be defined, for every $p, q \in \bigtriangleup(S)$, as  $$R(p\lVert q)=\int \log \Big( \dfrac{\textnormal{d} p }{\textnormal{d} q} \Big) \textnormal{d} p $$  if $p \ll q$; $+\infty$ otherwise. These preferences satisfy \nameref{ax_DAAA}; they are a special case of the act-dependent variational model of Proposition \ref{dependent_variational_main}.


We interpret the representation \eqref{multiplier_pref_oo} in light of the risk mitigation story.\footnote{The conflicting selves narrative would apply as well with obvious adjustments. We employ it later to interpret the model of Example \ref{confidence_benchmark} concerning \nameref{ax_DRAA}.} Two forces contribute to the evaluation of the act $f$. On the one hand, the decision maker's optimal action induces the most favorable model among the ones justifying $f$ over the outside option. This model, $q \in C_f$, determines the benchmark probability for the relative entropy $R(\cdot \lVert q)$. On the other hand, due to concerns for misspecification, the decision maker employs the worst model $p \in \bigtriangleup(S)$ to compute the expected utility, trading-off a higher value of the relative entropy $R(p \lVert q)$ the more $p$ diverges from $q$.\hfill $ \blacktriangleleft$
\end{example}


\section{Relative ambiguity attitudes}
We now study the impact of a \textit{relative} change in the proportion of the certainty part of an act. In particular, we characterize preferences displaying \nameref{ax_DRAA}. Recall that $x_*\in X$ denotes the worst consequence, that is,  $y\succsim x_*$ for all $y\in X$.

\begin{axiom}[decreasing relative ambiguity aversion]\label{ax_DRAA}
For all $f\in\mathcal{F}$, $y\in X$, and $\alpha,\beta\in (0,1)$, if $\alpha\leq \beta$, then \[ \alpha f+(1-\alpha)x_*\succsim \alpha y+(1-\alpha)x_*\implies \beta f+(1-\beta)x_*\succsim \beta y+(1-\beta)x_*.\]\end{axiom} 
This axiom says that if an act $\alpha f+(1-\alpha)x_*$ is preferred to a constant act $\alpha y+(1-\alpha)x_*$, where both acts can be expressed as mixtures with the worst consequence, then such a ranking is preserved after decreasing the proportion associated with the worst consequence in both acts. This axiom modifies the one proposed by \cite{Xuechanging} by restricting to mixtures involving the worst consequence only.

We are ready to state our main representation for \nameref{ax_DRAA}. It mirrors the representation of \nameref{ax_DAAA} in Theorem \ref{thm:UAPstyleDAAA}, the main distinctions being: \textit{(i)} constant superadditivity of the aggregators is replaced by positive superhomogeneity, and \textit{(ii)} we require the existence of a worst consequence.

\begin{theorem} \label{representation_IRAA}
Let $\succsim$ be a binary relation over $\mathcal{F}$. The following are equivalent
\begin{enumerate}[label=(\roman*)]
    \item \label{item1_rep_IRAA} $\succsim$ is MBA,  admits a worst consequence, and exhibits \nameref{ax_DRAA}. 
    \item \label{item2_rep_IRAA} There exist an affine function $u:X\to \mathbb{R}$ with $\min u(X)=0$ and a linearly continuous family $\mathcal{G}$ of monotone, positively superhomogeneous in the first argument, and quasiconvex functions $G:\mathbb{R}\times \bigtriangleup(S)\to \bar{\R}$ such that
    \[
    f\succsim g \Longleftrightarrow  \max_{G\in \mathcal{G}}\inf_{p\in \bigtriangleup(S)}G\left(\int u(f)\textnormal{d}p,p\right)\geq \max_{G\in \mathcal{G}}\inf_{p\in \bigtriangleup(S)}G\left(\int u(g)\textnormal{d}p,p\right)
    \]
    for all $f,g\in \mathcal{F}$ and $\max_{G\in \mathcal{G}}\inf_{p\in \bigtriangleup(S)}G\left(\int u(x)\textnormal{d}p,p\right)=u(x)$ for all $x \in X$.
\end{enumerate}
\end{theorem}

The two interpretations provided for  Lemma \ref{Lemma_1} 
 and Theorem \ref{thm:UAPstyleDAAA}---the dual-self and the risk mitigation interpretation---apply to Theorem \ref{representation_IRAA} as well. In particular, through the positive superhomogeneity of each $G(\cdot,p)$, for all $p\in \bigtriangleup(S)$, we can view \nameref{ax_DRAA} as a property of the aggregators only.

Analogous to Example \ref{example1DAAA}, the following connects the representation of Theorem \ref{representation_IRAA} to second-order expected utility with risk mitigation when the function governing ambiguity attitudes satisfies decreasing relative risk aversion (DRRA), i.e., the Arrow-Pratt coefficient of relative risk aversion is decreasing.

\begin{example}[second-order expected utility with risk mitigation; DRRA]\label{example1DRAA}
Consider again the preferences introduced in Example \ref{example1DAAA}. In particular, for $f \in \mathcal{F}$, $$V(f)=\max_{q \in Q} \> I_q(u(f))= \max_{q \in Q} \> \phi^{-1}\left(\int \phi(u(f))\textnormal{d}q\right)$$ where $Q \subseteq \bigtriangleup^{\sigma}(S)$ is a finite set of models, $u(X)=[0, \infty)$, and $\phi$ is a continuous, strictly increasing and concave function. By Lemma \ref{Lemma_constsuppossup_examples} and \ref{representation_for_examples} in Appendix \ref{Appendix_B3}, if $\phi$ is also twice differentiable and DRRA, i.e., $t\mapsto-t\phi''(t)/\phi'(t)$ is decreasing, then each $I_q$ is positively superhomogeneous.\footnote{For instance, the function $\phi:t\mapsto t+\sqrt{t}$ satisfies all the listed hypotheses.} As a result, these preferences satisfy \nameref{ax_DRAA} and can be represented following Theorem \ref{representation_IRAA}, where each $G_q$, defined as in Example \ref{example1DAAA}, is positively superhomogeneous.  \hfill $\blacktriangleleft$

\end{example}

\begin{remark} \label{concave_DRRA_smooth} We show that, under ambiguity aversion, our formulation of \nameref{ax_DRAA} is equivalent to the smooth ambiguity functional (\citeauthor{klibanoff2005smooth}, \citeyear{klibanoff2005smooth}) satisfying DRRA. Formally, for $f \in \mathcal{F}$, the smooth certainty equivalent is $$V(f)=\phi^{-1}\left(\int \phi\Big(\int u(f) \textnormal{d} p \Big) \textnormal{d} \mu\right)$$ where we assume that $u(X)=[1,\infty)$,  $\mu$ is a countably additive probability measure over $\bigtriangleup^{\sigma}(S)$, $(S,\Sigma)$ is a measurable space and $\Sigma$ is non-trivial, and $\phi: [1, \infty) \to \R$ is continuous, strictly increasing, and concave, capturing ambiguity aversion. Proposition \ref{prop_smooth_draa} in Appendix \ref{appendix_smooth_DRRA} shows that, if $\phi$ is twice differentiable, then smooth ambiguity preferences satisfy \nameref{ax_DRAA} for all countably additive probability measures $\mu$ over $\bigtriangleup^{\sigma}(S)$  if and only if $\phi$ is DRRA.   \hfill $ \blacktriangleleft$
\end{remark}


Next, we provide an alternative representation for \nameref{ax_DRAA} relating this property, which is weaker than worst independence (see section \ref{discussion_independence}), to an act-dependent confidence preference representation. In their standard formulation (\citeauthor{chateauneuf2009ambiguity}, \citeyear{chateauneuf2009ambiguity}), confidence preferences are characterized by a confidence function $d:\bigtriangleup(S)\to [0,\infty]$ associating each probability model to its relative confidence level. In our case, as in Proposition \ref{dependent_variational_main}, due to changing ambiguity attitudes and the lack of convexity, the set of confidence functions varies for each act.

\begin{proposition} \label{dependent_confidence_main} Let $\succsim$ be a binary relation over $\mathcal{F}$. The following are equivalent
\begin{enumerate}[label=(\roman*)]
    \item $\succsim$ is MBA, admits a worst consequence, and exhibits \nameref{ax_DRAA}.
    \item There exist an affine function $u:X\to \mathbb{R}$, with $\min u(X)=0$, and, for all $f\in \mathcal{F}$, a set $D_{f}$ of upper semicontinuous and quasiconcave $d:\bigtriangleup(S)\to [0,\infty]$ such that
\[
f \succsim g \iff \max\limits_{d\in D_{f}}\min\limits_{p\in \bigtriangleup(S)}\frac{\int_{S}u(f)\textnormal{d}p}{d(p)} \geq \max\limits_{d\in D_{g}}\min\limits_{p\in \bigtriangleup(S)}\frac{\int_{S}u(g)\textnormal{d}p}{d(p)}
\]
where $\max_{d\in D_{x}}\min_{p\in \bigtriangleup(S)} u(x) / d(p)=u(x)$ for all $x\in X$, and $D_{f}\subseteq D_{g}$ for all $f,g\in \mathcal{F}$ with $u(f)\leq u(g)$.
\end{enumerate}
\end{proposition}

This result generalizes the existing representations for homothetic preferences mentioned above and mirrors the act-dependent variational model of Proposition \ref{dependent_variational_main}. Notice, positive superhomogeneity of the aggregators (Theorem \ref{representation_IRAA}) is captured by the set of confidence functions enlarging as the utility levels increase. Therefore, the decision maker displays a higher degree of relative confidence for higher
utility levels. 




The following example relates to the representation of \nameref{ax_DRAA} of Proposition \ref{dependent_confidence_main} by considering act-dependent entropic-confidence preferences with outside option.



\begin{example}[confidence preferences with outside option]  \label{confidence_benchmark} Following Example \ref{multiplier_benchmark}, we model a decision maker that can mitigate the uncertainty of each act $f \in \mathcal{F}$ by exerting effort whenever it is valuable to do so. Formally, to each act $f \in \mathcal{F}$ is associated the set
\[
D_f=\left\lbrace q\in Q:\int u(f)\textnormal{d} q\geq \theta \right\rbrace\cup\left\lbrace q_{\text{u}}\right\rbrace
\]
where $Q \subseteq \bigtriangleup(S)$ is a finite set of models, $S$ is finite, $\Sigma=2^S$, $\theta\in [0, \infty)$, and $q_{\text{u}}$ denotes the uniform distribution over $S$. Notice that, for $f, g \in \mathcal{F}$, $u(f) \leq u(g)$ implies $D_f \subseteq D_g$. 

We define act-dependent entropic-confidence preferences as follows \begin{equation} \label{act_dep_entropic_confidence}
 V(f)=\max\limits_{q\in D_f}\min\limits_{p\in \bigtriangleup(S)}\frac{\int u(f)\textnormal{d} p }{ \textnormal{exp}(-R(p\lVert q))}   
\end{equation}
where, as before, $R(\cdot\lVert\cdot)$ denotes the relative entropy. These preferences satisfy \nameref{ax_DRAA}; they are a special case of the act-dependent confidence preferences of Proposition \ref{dependent_confidence_main}.

We interpret the representation \eqref{act_dep_entropic_confidence} by following the narrative of the intrapersonal game between two conflicting selves. To evaluate the act $f$, Optimism selects the most favorable model among the ones justifying $f$ over the outside option. This choice, $q \in D_f$, determines the benchmark probability to compute the relative entropy $R(\cdot \lVert q)$. As a result, Pessimism faces a trade-off: the choice of the pessimistic model $p \in \bigtriangleup(S)$ determines the expected utility of the agent, but, at the same time, the more it diverges from $q$, the higher the value of the relative entropy $R(p \lVert q)$, and of the ratio in \eqref{act_dep_entropic_confidence}. \hfill $ \blacktriangleleft$
\end{example}

\section{A risk sharing application}
In this section, we apply our representations of changing ambiguity attitudes to an exchange economy with a single consumption good and no aggregate uncertainty. We want to investigate whether it is efficient for agents displaying general ambiguity preferences to take bets. We show that, under regularity requirements, the joint combination of \nameref{ax_DAAA} and increasing relative ambiguity aversion, when at least one of them is non-constant, implies \textit{strict pseudoconcavity at certainty}. Under this condition, as shown by \cite{ghirardato2018risk}, betting is inefficient if and only if the agents share at least one \textit{supporting probability},\footnote{We adopt the terminology of \cite{ghirardato2018risk}, although \cite{rigotti2008subjective} first introduced this notion under the name of \textit{subjective beliefs}.} i.e., the agents' beliefs that support the consumption bundle have a non-empty intersection.

\paragraph{Non-constant changing ambiguity attitudes.} We introduce non-constant absolute and relative ambiguity attitudes, which
follow from the notions already discussed with obvious adjustments.

\begin{axiom}[non-constant decreasing absolute ambiguity aversion]\label{ax_nc_DAAA}
For all $f\in\mathcal{F}$, $x,y,z\in X$, and $\alpha\in (0,1)$, if $y\succ x$, then \[
\alpha f+(1-\alpha)x\succsim \alpha z+(1-\alpha)x\implies \alpha f+(1-\alpha)y\succ \alpha z+(1-\alpha)y.
\] \end{axiom}

Recall, if it exists, $x_*\in X$ denotes the worst consequence, that is,  $y\succsim x_*$ for all $y\in X$.

\begin{axiom}[non-constant increasing relative ambiguity aversion]\label{ax_nc_IRAA}
For all $f\in\mathcal{F}$, $y\in X$, and $\alpha,\beta\in (0,1)$, if $\alpha > \beta$, then
\[
\alpha f+(1-\alpha)x_*\succsim \alpha y+(1-\alpha)x_*\implies \beta f+(1-\beta)x_*\succ \beta y+(1-\beta)x_*.
\]  \end{axiom}

Lemma \ref{lemma:combined} in Appendix \ref{appendix_risk_sharing} shows that, under MBA preferences, \nameref{ax_DAAA} and increasing relative ambiguity aversion hold simultaneously if and only if the certainty equivalent functional satisfies both constant superadditivity and positively subhomogeneity.\footnote{In the context of recursive ambiguity models, Lemma 1 in \cite{strzalecki2013temporal} shows that constant superadditive and positively subhomogeneous functionals represent preferences for early resolution of uncertainty.} Furthermore, whenever either one of the two behavioral properties holds non-constantly, the corresponding functional property holds strict.



\paragraph{The economy.} We model an exchange economy populated by finitely many agents $N:=\{1, \dots, N\}$. Let the state space $S$ be finite as well. Each agent, indexed by $i \in N$, has a utility function  $V_i:\mathbb{R}_{+}^S\rightarrow\mathbb{R}$ and is dispensed with an endowment $\omega_i \in \mathbb{R}_{+}^S$. For simplicity, we abstract away from risk attitudes and restrict the attention to the case of risk neutrality, i.e., $V_i(x)=x$ for every $x\in\mathbb{R}_+$. Finally, as in \cite{ghirardato2018risk}, this economy features no aggregate uncertainty, i.e., $\sum_i \omega_i=\bar{\omega}$ for some $\bar{\omega}>0$.

An \textit{allocation} is a vector $(f_1, \ldots, f_N)\in \mathbb{R}_{+}^{N\times S}$ where each $f_i $ is the consumption bundle assigned to agent $i \in N$ contingent on each state. We say that an allocation is \textit{feasible} if $\sum_i f_i=\bar{\omega}$; \textit{interior} if, for all $i \in N$, $f_i >0$; \textit{full-insurance} if, for all $i \in N$, $f_i= x_i$ for some $x_i \in \mathbb{R}_{+}$; \textit{Pareto-efficient} if it is feasible, and there is no other feasible allocation $(g_1, \ldots, g_N)$ such that $V_i(g_i) \geq V_i(f_i)$ for all $i \in N$, and $V_j(g_j) > V_j(f_j)$ for some $j \in N$; a \textit{competitive equilibrium with transfers} if it is feasible, and there exist prices $q\in \mathbb{R}_{++}^S$ and transfers $(T_i)_{i \in N}\in \mathbb{R}_{+}^N$  with $\sum_{i} T_i=0$ 
 such that $f_i \in \arg\max_{\{g\in \mathbb{R}_{+}^S:q \cdot g\leq q \cdot \omega_i +T_i\}}V_i(g)$ for all $i \in N$.
 

\paragraph{Supporting probabilities and strict pseudoconcavity.} For every $i \in N$, the set of \textit{supporting probabilities} at an allocation $f\in \R^S_+$ is $$ \pi_i(f)=\left\{p \in \bigtriangleup(S): \forall g \in \mathbb{R}_{+}^S, V_i(g) \geq V_i(f) \implies p \cdot g \geq p \cdot f\right\}.$$ Notice that, $\pi_i(f)$ can be interpreted as the set of (normalized) prices such that any bundle weakly preferred to $f$ is at least as costly as $f$.


Following \cite{ghirardato2018risk}, a function $V:\R^S_+\to \R$ is \textit{strictly pseudoconcave at} $f\in \R^S_{++}$ if, for all $g\neq f$,
\[
V(g)\geq V(f)\Longrightarrow \forall q\in \partial V(f), \>\>\> q\cdot (g-f)>0,
\]
where $\partial V(f)$ denotes the  Clarke subdifferential of $V$ at $f$ (\citeauthor{clarke}, \citeyear{clarke}).\footnote{For $n \geq 1$, an open subset $B \subseteq \R^n$, and a function $V : B \to \R$, the \textit{Clarke subdifferential} of $V$ at $b \in B$ is
$$\partial V(b) = \text{cl} \> \text{conv} \left\{
\lim_{k\to \infty} d^k:\exists (b^k) \to b \> \text{such that} \> d^k = \nabla V(b^k), \forall k\right\}, 
$$ where $\text{cl} \> \text{conv}$ denotes the closure of the convex hull, and $\nabla V(b^k)$ the gradient of $V$ at $b^k$.} 
Furthermore, $V$ is \textit{strict pseudoconcavity at certainty} if it is strictly pseudoconcave at $x$ for all $x \in \R_{++}$.    
\paragraph{Changing ambiguity attitudes and risk sharing.}  As in \cite{ghirardato2018risk}, we say that $V: \mathbb{R}_{+}^S \rightarrow \mathbb{R}$ is \textit{nice} if it is locally Lipschitz, strictly monotone and, for every $x \in \R_{++}$, continuously differentiable in a neighborhood of $x$ with $\nabla V(x)\neq 0$, where $\nabla V(x)$ denotes the gradient of $V$ at $x$. We are ready to state the main result of this section.


\begin{proposition}\label{risksharingth}
 If  $V:\mathbb{R}_{+}^S\rightarrow\mathbb{R}$ is normalized, nice, and satisfies constant superadditivity and positive subhomogeneity, with at least one being strict, then $V$ satisfies strict pseudoconcavity at certainty.
\end{proposition}

The proof follows from observing that the joint combination of constant superadditivity and positive subhomogeneity implies that $V$ satisfies \virgo{concavity at certainty,} that is, for every $f \in \mathcal{F}$, $x \in \R_{++}$, and $\alpha \in (0,1)$, $$V(\alpha f + (1-\alpha) x) \geq \alpha V(f) + (1-\alpha) V(x).$$

By applying Proposition \ref{risksharingth} and Theorem 3 in \cite{ghirardato2018risk}, we retrieve the following equivalences, which generalize the risk sharing result of \cite{rigotti2008subjective}. Thus, we establish a connection between risk sharing and our general representations of changing ambiguity attitudes.

\begin{corollary}\label{corollary_NiceAllocations} For each $i \in N$, assume $V_i$ is nice, and satisfies constant superadditivity and positive subhomogeneity, with at least one being strict. The following are equivalent
\begin{enumerate}
\item[(i)]  Every Pareto-efficient allocation is a full-insurance allocation.
\item[(ii)] Every feasible, full-insurance allocation is Pareto-efficient.
\item[(iii)]  For every feasible, full-insurance allocation $(x_1, \ldots, x_N)$, $\>\> \bigcap_{i\in N} \pi_i( x_i) \neq \emptyset.$
\end{enumerate}
Furthermore, under the above equivalent conditions, every interior, feasible, full-insurance allocation is a competitive equilibrium with transfers. 
\end{corollary}



In contrast to this result, \nameref{ax_DRAA} appears at odds with risk sharing. Intuitively, by Theorem \ref{representation_IRAA}, this property implies that preferences are \virgo{convex at zero,} that is, for every $f \in \mathcal{F}$ and $\alpha\in (0,1)$, $$V(\alpha f+(1-\alpha)0)\leq \alpha V(f)+(1-\alpha)V(0).$$  The following example considers an economy where each $V_i$ is nice but satisfies positive superhomogeneity as well as constant additivity. We construct a full-insurance allocation that is not Pareto-efficient. This suggests that agents displaying \nameref{ax_DRAA} may be willing to bet even if they share the same beliefs.

\begin{example} \label{DRAA_betting}
Let $S=\{s_1, s_2\}$, $N=\{1,2\}$, and $\omega_1=\omega_2=(1/2, 1/2)$. Define $V_1=V_2:=V$ as 
$$V(f)=\max_{H\in\Psi} H(f),$$
where $\Psi=\{H_1,H_2,H_3\}$, and \begin{align*}
H_1(f)&=1/9 \cdot f(s_1)+8/9 \cdot f(s_2)-0.1 \\
H_2(f)&=\frac{1}{10}\log\left(\frac{1}{2}\left(e^{10\left(\frac{1}{4}f(s_1) + \frac{3}{4}f(s_2)\right)}\right) + \frac{1}{2}\left(e^{10\left(\frac{3}{4}f(s_1) + \frac{1}{4}f(s_2)\right)}\right)\right)\\
H_3(f)&=8/9 \cdot f(s_1) + 1/9 \cdot f(s_2)-0.1.
\end{align*}
\noindent
Since $V_1=V_2$ the agents share the same beliefs at the initial full-insurance allocation. However, we show that such allocation is not Pareto-efficient which implies that Corollary \ref{corollary_NiceAllocations} does not hold under non-constant \nameref{ax_DRAA}. To this end, notice that, $V$ is monotone, normalized, and continuous. Therefore, by Proposition 1 in \cite{OmnibusSimone}, $V$ represents MBA preferences. Furthermore, since each $H\in\Psi$ is positively superhomogeneous, $V$ is positively superhomogeneous as well, and, by Theorem \ref{representation_IRAA}, \nameref{ax_DRAA} holds. Finally, $V$ is also nice.



The initial full-insurance endowment is not Pareto-efficient. Indeed, the feasible allocation $((0.4,0.6),(0.6,0.4))$ achieves a strictly higher level of utility for both agents:\begin{align*}
H_1((1/2,1/2)) = 0.4,\quad
H_2((1/2,1/2)) =  0.5, \quad
H_3((1/2,1/2)) = 0.4,
\end{align*}
while \begin{align*}
H_1((0.4,0.6)) &= H_3((0.6,0.4)) =  0.4 \overline{7}\\
H_2((0.4,0.6)) &= H_2((0.6,0.4)) = 0.512 \\
H_3((0.4,0.6)) &= H_1((0.6,0.4)) =0.3 \overline{2},
\end{align*} which implies that $V_1((0.4,0.6))=V_2((0.6,0.4))=0.512> V((1/2,1/2))=0.5.$ \hfill $ \blacktriangleleft$
\end{example}

\appendix

\section*{Appendix}
\addcontentsline{toc}{section}{Appendices}
\renewcommand{\thesubsection}{\Alph{subsection}}

The Appendix is organized as follows. In Appendix \ref{app_A}, we present the proofs of some mathematical results we later employ in Appendix \ref{app_B} to prove the results in the main text. 

\subsection{Mathematical appendix} \label{app_A}

\subsubsection{Toolkit lemmas}

In this part of the appendix we provide characterizations of constant superadditivity which we later use to prove our results in the main text. 

\begin{lemma}\label{banale}
Fix $K\subseteq \R$. A map $I:\BK\to \R$ is constant superadditive if and only if $I(\varphi+k)\leq I(\varphi)+k$ for all $\varphi\in \BK$ and $k\leq 0$ such that $\varphi+k\in \BK$.
\end{lemma}
\begin{proof}
Suppose $I$ is constant superadditive, $\varphi\in \BK$, and $k\leq 0$ such that $\varphi+k\in \BK$. It follows that
\[
I(\varphi)=I(\varphi+k-k)\geq I(\varphi+k)-k.
\]
The converse is proved in an analogous fashion.
\end{proof}

\begin{lemma}\label{setticemia}
Fix a convex $K\subseteq \R$ with $0\in \textnormal{int}K$ and a continuous map $I:\BK\to \R$. We have that $I$ is constant superadditive if and only if 
\begin{equation}\label{setticemia_cond1}
I(\alpha \varphi+(1-\alpha)k)\geq I(\alpha\varphi)+(1-\alpha)k
\end{equation}
for all $\alpha\in [0,1]$, $\varphi\in \BK$, and $k\in K\cap \mathbb{R}_+$.
\end{lemma}
\begin{proof}
The proof follows from a minor modification of Lemma 5 in \cite{Niveloidsextension}, we report it here for completeness. If $I$ is constant superadditive, then it is straightforward to see that \eqref{setticemia_cond1} must hold. As for the converse, first notice that, since $I$ is continuous, we can assume that $K=(a,b)$ for some $a,b\in \bar{\R}$ without loss of generality. Let $\varphi\in B_0(\Sigma,K)$ and $k\in \mathbb{R}_{++}$ be such that $\varphi+k\in \BK$. Then, we have that $a<\inf \varphi,\sup(\varphi+k)<b$ and there exists $\alpha\in (0,1)$ such that $\varphi/\alpha,(\varphi+k)/\alpha\in \BK$. Since $k>0$ and $0\in \textnormal{int}K=(a,b)$, there exists $n\geq 2$ such that
\[
\frac{\frac{1}{n}k}{1-\alpha}\in K.
\]
Then, by convexity of $K$,
\[
\frac{\varphi+\frac{m}{n}k}{\alpha}=\frac{m}{n}\frac{\varphi+k}{\alpha}+\left(1-\frac{m}{n}\right)\frac{\varphi}{\alpha}\in \BK
\]
for all $m=0,\ldots,n-1$. For all $m=0,\ldots,n-1$, by \eqref{setticemia_cond1}, it follows that
\[
I\left(\varphi+\frac{m+1}{n}k\right)=I\left(\alpha\frac{\varphi+\frac{m}{n}k}{\alpha}+(1-\alpha)\frac{\frac{1}{n}k}{1-\alpha}\right)\geq I\left(\varphi+\frac{m}{n}k\right)+\frac{k}{n}.
\]
Moreover, this yields,
\[
I\left(\varphi+k\right)-I(\varphi)=\sum_{m=0}^{n-1}\left(I\left(\varphi+\frac{m+1}{n}k\right)-I\left(\varphi+\frac{m}{n}k\right)\right)\geq \sum_{m=0}^{n-1}\frac{k}{n}=k.
\]
Thus, $I$ is constant superadditive.
\end{proof}

For all $T:\BK\to \R$,  $\varphi \in \BK$, define $\bar{T}: B_0(\Sigma,-K)\to \R$ as $\bar{T}(\varphi)=-T(-\varphi)$.

\begin{lemma}\label{banalino}
Fix $K\subseteq \R$ and a map $I:\BK\to \R$. Then,  $I$ is constant superadditive if and only if $\bar{I}$ is constant superadditive.
\end{lemma}
\begin{proof}
Suppose $I$ is constant superadditive. If $\varphi\in B_0(\Sigma,-K)$, $k\leq 0$, and $\varphi+k\in B_0(\Sigma,-K)$, then
\[
\bar{I}(\varphi+k)=-I(-\varphi-k)\leq -I(-\varphi)+k=\bar{I}(\varphi)+k.
\]
Since $k\leq 0$, by Lemma \ref{banale}, $\bar{I}$ is constant superadditive. The converse follows by $I=\bar{\bar{I}}$.
\end{proof}

\subsubsection{Extension results}

Here we provide some instrumental extension results. For each interval $K\subseteq \R$, we denote by $K_{\infty}:=K\cup [\sup K,\infty)$.

\begin{lemma}\label{const_superadd_extension}
If $K$ is an interval and $T:\BK\to \R$ is constant superadditive, monotone, and normalized, then the map $\tilde{T}:B_0(\Sigma,K_{\infty})\to \bar{\R}$ defined as
\[
\tilde{T}(\psi)=\sup\left\lbrace T(\varphi)+m:\varphi\in \BK,\ m\geq 0,\ \varphi+m\leq \psi \right\rbrace
\]
for all $\psi\in B_0(\Sigma,K_{\infty})$ is a real-valued, constant superadditive, monotone, and normalized extension of $T$ to $B_0(\Sigma,K_{\infty})$.
\end{lemma}
\begin{proof} 
Monotonicity is immediate as the larger $\psi$ gets the larger is the set over which we are taking the supremum. To prove that $\tilde{T}$ extends $T$ notice that for all $\psi\in \BK$, we have $\tilde{T}(\psi)\geq T(\psi)$. Conversely, notice that for all $\psi,\varphi\in \BK$, and $m\geq 0$, with $\varphi+m\leq \psi$, we have that $\varphi+m\in \BK$ and by monotonicity and constant superadditivity of $T$, we have
\[
T(\varphi)+m\leq T(\varphi+m)\leq T(\psi).
\]
By the arbitrariness of $\varphi$ and $m$, we have $\tilde{T}(\psi)=T(\psi)$. Thus, $\tilde{T}$ is a monotonic extension of $T$. Now we prove that $\tilde{T}$ satisfies constant superadditivity. To this end notice first that
\[
\tilde{T}(\psi)=\sup\limits_{\varphi\in \BK,\ \varphi\leq \psi}\left\lbrace T(\varphi)+\inf\limits_{s\in S}\left\lbrace \psi(s)-\varphi(s) \right\rbrace\right\rbrace
\]
for all $\psi\in B_0(\Sigma,K_{\infty})$. Let $\psi\in B_0(\Sigma,K_{\infty})$ and $k\geq 0$,
\begin{align*}
\tilde{T}(\psi+k)&=\sup\limits_{\varphi\in \BK,\ \varphi\leq \psi+k}\left\lbrace T(\varphi)+\inf\limits_{s\in S}\left\lbrace \psi(s)+k-\varphi(s) \right\rbrace\right\rbrace\\
&\geq \sup\limits_{\varphi\in \BK,\ \varphi\leq \psi}\left\lbrace T(\varphi)+\inf\limits_{s\in S}\left\lbrace \psi(s)+k-\varphi(s) \right\rbrace\right\rbrace\\
&=\sup\limits_{\varphi\in \BK,\ \varphi\leq \psi}\left\lbrace T(\varphi)+\inf\limits_{s\in S}\left\lbrace \psi(s)-\varphi(s) \right\rbrace\right\rbrace+k\\
&=\tilde{T}(\psi)+k.
\end{align*}
Next we prove that $\tilde{T}$ is a real-valued map. If $\psi\in \BK$, then $\tilde{T}(\psi)=T(\psi)\in \R$. If $\psi\in B_0(\Sigma,K_{\infty})\setminus \BK$, then for all $\varphi\in \BK$ and $m\geq 0$ with $\varphi+m\leq \psi$, we have that
\[
T(\varphi)+m\leq T(\sup \varphi)+m=\sup\varphi+m\leq \sup \psi<\infty
\]
the equality follows from the normalization of $T$ and the fact that $\varphi$ is a simple function.

To conclude we show that $\tilde{T}$ is also normalized. If $k\in K$, then $\tilde{T}(k)=T(k)=k$. If $k\in K_{\infty}\setminus K$, then $k\geq \sup K$. Then, there exists $t\in K$ and $m\geq 0$ such that $t+m=k$, and hence $\tilde{T}(k)=T(t)+m=t+m=k$. Thus, $\tilde{T}$ is normalized.
\end{proof}

\begin{lemma}\label{positive_superhomo_extension}
If $K$ is an interval with $\inf K=0$ and $T:\BK\to \R$ is positively superhomogeneous, monotone, and normalized, then the map $\tilde{T}:B_0(\Sigma,K_{\infty})\to \bar{\R}$ defined as
\[
\tilde{T}(\psi)=\sup\left\lbrace \alpha T(\varphi):\varphi\in \BK,\ \alpha\geq 1,\ \alpha\varphi\leq \psi \right\rbrace
\]
for all $\psi\in B_0(\Sigma,K_{\infty})$ is a positively superhomogeneous, monotone, and normalized extension of $T$ to $B_0(\Sigma,K_{\infty})$.
\end{lemma}
\begin{proof}
Monotonicity is immediate as the larger $\psi$ gets the larger is the set over which we are taking the supremum. To prove that $\tilde{T}$ extends $T$ notice that for all $\psi\in \BK$, we have $\tilde{T}(\psi)\geq T(\psi)$. Conversely, notice that for all $\psi,\varphi\in \BK$, and $\alpha\geq 1$, with $\alpha\varphi\leq \psi$, we have that $\alpha \varphi\in \BK$ and by monotonicity and positive superhomogeneity of $T$, we have
\[
\alpha T(\varphi)\leq T(\alpha\varphi)\leq T(\psi).
\]
By the arbitrariness of $\varphi$ and $\alpha$, we have $\tilde{T}(\psi)=T(\psi)$. Thus, $\tilde{T}$ is a monotonic extension of $T$. Now we prove that $\tilde{T}$ satisfies positive superhomogeneity. To this end notice first that
\[
\tilde{T}(\psi)=\sup\limits_{\varphi\in \BK,\ \varphi\leq \psi}\left\lbrace T(\varphi)\inf\limits_{s\in S:\varphi(s)>0}\left\lbrace \frac{\psi(s)}{\varphi(s)} \right\rbrace\right\rbrace
\]
for all $\psi\in B_0(\Sigma,K_{\infty})$. Let $\psi\in B_0(\Sigma,K_{\infty})$ and $\alpha\geq 1$,
\begin{align*}
\tilde{T}(\alpha\psi)&=\sup\limits_{\varphi\in \BK,\ \varphi\leq \alpha\psi}\left\lbrace T(\varphi)\inf\limits_{s\in S:\varphi(s)>0}\left\lbrace \frac{\alpha\psi(s)}{\varphi(s)} \right\rbrace\right\rbrace\\
&\geq \sup\limits_{\varphi\in \BK,\ \varphi\leq \psi}\left\lbrace T(\varphi)\inf\limits_{s\in S:\varphi(s)>0}\left\lbrace \frac{\alpha\psi(s)}{\varphi(s)} \right\rbrace\right\rbrace\\
&=\alpha\sup\limits_{\varphi\in \BK,\ \varphi\leq \psi}\left\lbrace T(\varphi)\inf\limits_{s\in S:\varphi(s)>0}\left\lbrace \frac{\psi(s)}{\varphi(s)} \right\rbrace\right\rbrace\\
&=\alpha\tilde{T}(\psi).
\end{align*}
Next we prove that $\tilde{T}$ is a real-valued map. If $\psi\in \BK$, then $\tilde{T}(\psi)=T(\psi)\in \R$. If $\psi\in B_0(\Sigma,K_{\infty})\setminus \BK$, then for all $\varphi\in \BK$ and $\alpha\geq 1$ with $\alpha \varphi\leq \psi$, we have that
\[
\alpha T(\varphi)\leq \alpha T(\sup \varphi)=\alpha\sup\varphi\leq \sup \psi<\infty
\]
the equality follows from the normalization of $T$ and the fact that $\varphi$ is a simple function.

To conclude we prove that $\tilde{T}$ is normalized. If $k\in K$, then $\tilde{T}(k)=T(k)=k$. If $k\in K_{\infty}\setminus K$, then $k\geq \sup K$. Therefore, there exist $\alpha\geq 1$ and $t\in K$ such that $\alpha t=k$, and hence $\tilde{T}(k)=\alpha T(t)=\alpha t=k$.
\end{proof}

\subsubsection{Envelope representations}
Before going into further results, we provide notation used throughout the remainder of Appendix \ref{app_A}. For a convex set $K\subseteq \mathbb{R}$, and $\xi,\varphi\in \BK$, we define the following sets
\begin{align*}
&C_{\xi}(\varphi)=\left\lbrace k\in \R:\xi+k\geq \varphi,\ \xi+k\in \BK \right\rbrace, \
M_{\xi}(\varphi)=\left\lbrace \alpha>0:\alpha\xi\geq \varphi,\ \alpha\xi\in \BK  \right\rbrace
\\
&c_{\xi}(\varphi)=\left\lbrace k\in \R:\varphi\geq \xi+k,\ \xi+k\in \BK \right\rbrace,\ m_{\xi}(\varphi)=\left\lbrace \alpha>0:\varphi\geq \alpha\xi,\ \alpha\xi\in \BK \right\rbrace.
\end{align*}
All the sets defined above could well be empty.
Fixing a map $I:\BK\to \bar{\R}$, we define the following auxiliary functionals associated to $I$:
 \begin{align*}
&I_{\xi}(\varphi)=\inf\left\lbrace I(\xi+k):k\in C_{\xi}(\varphi) \right\rbrace, \ J_{\xi}(\varphi)=\inf\left\lbrace I(\alpha\xi):\alpha\in M_{\xi}(\varphi) \right\rbrace,\\ 
& S_{\xi}(\varphi)=\sup\left\lbrace I(\xi+k):k\in c_{\xi}(\varphi) \right\rbrace, H_{\xi}(\varphi)=\sup\left\lbrace I(\alpha\xi):\alpha\in m_{\xi}(\varphi) \right\rbrace
\end{align*} for all $\xi,\varphi\in \BK$. We adopt the convention that $\inf\emptyset=\infty$ and $\sup\emptyset=-\infty$. If $C_{\xi}(\varphi)=M_{\xi}(\varphi)=\emptyset$, then $I_{\xi}(\varphi)=J_{\xi}(\varphi)=\infty$, while if $c_{\xi}(\varphi)=m_{\xi}(\varphi)=\emptyset$, $S_{\xi}(\varphi)=H_{\xi}(\varphi)=-\infty$.

\begin{proposition}\label{auxiliary_functionals_basic}
Fix a convex $K\subseteq \R$ and a monotone and normalized map $I:\BK\to \bar{\R}$. Then, for all $\xi\in \BK$,
\begin{enumerate}[label=(\roman*)]
\item $I_{\xi}$, $J_{\xi}$, $S_{\xi}$, and $H_{\xi}$ are monotone.
\item $I_{\xi}\geq I$, $J_{\xi}\geq I$, $S_{\xi}\leq I$, and $H_{\xi}\leq I$. 
\item $I_{\xi}$ is quasiconvex and $S_{\xi}$ is quasiconcave.
\item If $K\subseteq \mathbb{R}_+$, then $J_{\xi}$ is quasiconvex and $H_{\xi}$ is quasiconcave.
\end{enumerate}
\end{proposition}
\begin{proof}
\begin{enumerate}[label=\textit{(\roman*)}]
\item For all $\varphi,\psi,\xi\in \BK$ with $\varphi\geq \psi$ we have $C_{\xi}(\varphi)\subseteq C_{\xi}(\psi)$, $M_{\xi}(\varphi)\subseteq M_{\xi}(\psi)$, $c_{\xi}(\psi)\subseteq c_{\xi}(\varphi)$, and $m_{\xi}(\psi)\subseteq m_{\xi}(\varphi)$. This implies that $I_{\xi}$, $J_{\xi}$, $S_{\xi}$ and $H_{\xi}$ are monotone.  

\item Fix $\varphi,\xi\in \BK$. If $C_{\xi}(\varphi)=\emptyset$, $I_{\xi}(\varphi)=\infty\geq I(\varphi)$, while $J_{\xi}(\varphi)=\infty\geq I(\varphi)$ whenever $M_{\xi}(\varphi)=\emptyset$, and similarly if $c_{\xi}(\varphi)=\emptyset$, then $I(\varphi)\geq -\infty=S_{\xi}(\varphi)$, while $I(\varphi)\geq -\infty=H_{\xi}(\varphi)$ whenever $m_{\xi}(\varphi)=\emptyset$. Thus, we assume such sets to be nonempty. By monotonicity of $I$, we have that $I(\xi+k)\geq I(\varphi)$ for all $k\in C_{\xi}(\varphi)$. Analogously, $I(\alpha\xi)\geq I(\varphi)$ for all $\alpha\in M_{\xi}(\varphi)$. Thus, $I_{\xi}(\varphi)\geq I(\varphi)$ and $J_{\xi}(\varphi)\geq I(\varphi)$ for all $\varphi\in \BK$. An analogous argument for $k\in c_{\xi}(\varphi)$ and $\alpha\in m_{\xi}(\varphi)$ concludes the proof of the claim.

\item Fix $\xi\in \BK$, $\alpha\in (0,1)$, and $\varphi_1,\varphi_2\in \BK$. If $\max\left\lbrace I_{\xi}(\varphi_1),I_{\xi}(\varphi_2)\right\rbrace=\infty$, then $
I_{\xi}(\alpha\varphi_1+(1-\alpha)\varphi_2)\leq \max\left\lbrace I_{\xi}(\varphi_1),I_{\xi}(\varphi_2)\right\rbrace$.
Thus, assume that $I_{\xi}(\varphi_1),I_{\xi}(\varphi_2)$ are finite. We prove that for all $t\in \mathbb{R}$,\footnote{Here we follow the same steps provided in \cite{RuoduCashSubad} (proof of Theorem 4.1 therein).}
\[
\left(\forall i=1,2,\ I_{\xi}(\varphi_i)\leq t\right)\Longrightarrow I_{\xi}\left(\alpha\varphi_1+(1-\alpha)\varphi_2\right)\leq t.
\]
For all $\varepsilon>0$, there exist $k_1,k_2\in \mathbb{R}$ such that, for all $i=1,2$, $\xi+k_i\geq \varphi_i$, $\xi+k_i\in \BK$, and
\[
I(\xi+k_i)\leq I_{\xi}(\varphi_i)+\varepsilon\leq t+\varepsilon.
\]
We have that $\alpha \varphi_1+(1-\alpha)\varphi_2\leq \alpha (\xi+k_1)+(1-\alpha)(\xi+k_2)\leq \xi+k$
where $k=\max\left\lbrace k_1,k_2 \right\rbrace$ and hence, by monotonicity of $I$,
\begin{align*}
I_{\xi}(\alpha\varphi_1+(1-\alpha)\varphi_2)&\leq I(\xi+k)\leq t+\varepsilon.
\end{align*}
The claim follows by arbitrariness of $\varepsilon$. The quasiconcavity of $S_{\xi}$ is proved analogously.
\item Fix $\xi\in \BK$, $\alpha\in (0,1)$, and $\varphi_1,\varphi_2\in \BK$. If $\max\left\lbrace J_{\xi}(\varphi_1),J_{\xi}(\varphi_2)\right\rbrace=\infty$, then $
J_{\xi}(\alpha\varphi_1+(1-\alpha)\varphi_2)\leq \max\left\lbrace J_{\xi}(\varphi_1),J_{\xi}(\varphi_2)\right\rbrace$.
Thus, assume that $J_{\xi}(\varphi_1),J_{\xi}(\varphi_2)$ are finite

Now we pass to the quasiconvexity of $J_{\xi}$. We prove that for all $t\in \mathbb{R}$,
\[
\left(\forall i=1,2,\ J_{\xi}(\varphi_i)\leq t\right)\Longrightarrow J_{\xi}\left(\alpha\varphi_1+(1-\alpha)\varphi_2\right)\leq t.
\]
For all $\varepsilon>0$, there are $\beta_1,\beta_2>0$ so that, for all $i=1,2$, $\beta_i\xi\geq \varphi_i$, $\beta_i\xi\in \BK$, and
\[
I(\beta_i\xi) \leq J_{\xi}(\varphi_i)+\varepsilon \leq t+\varepsilon.
\]
Since $K\subseteq \R_+$, we have that $\xi\geq0$ and hence $\alpha \varphi_1+(1-\alpha)\varphi_2\leq \alpha \beta_1\xi +(1-\alpha) \beta_2\xi \leq \beta\xi$
where $\beta=\max\left\lbrace \beta_1,\beta_2 \right\rbrace$. Therefore, by monotonicity of $I$,
\begin{align*}
J_{\xi}(\alpha\varphi_1+(1-\alpha)\varphi_2)&\leq I(\beta \xi)\leq t+\varepsilon.
\end{align*}
By the arbitrariety of $\varepsilon$ the claim follows. The quasiconcavity of $H_{\xi}$ is proved analogously.
\end{enumerate}
\end{proof}

\begin{proposition}\label{auxiliary_functionals_constant_sup}
Fix a convex $K\subseteq \R$ and a monotone map $I:\BK\to \bar{\R}$. If $I$ is constant superadditive, then, for all $\xi\in \BK$, $I_{\xi}$ and $S_{\xi}$ are constant superadditive.
\end{proposition}
\begin{proof}
Let $\varphi,\xi\in \BK$, and $m\geq 0$ be such that $\varphi+m\in \BK$. If $I_{\xi}(\varphi+m)=\infty$, then the claim follows. Thus, suppose $I_{\xi}(\varphi+m)\in \R$. By Proposition \ref{auxiliary_functionals_basic} and the monotonicity of $I$, we have that $I_{\xi}$ is monotone, and hence $I_{\xi}(\varphi)\in \mathbb{R}$. These simple observations imply that $C_{\xi}(\varphi+m)$ and $C_{\xi}(\varphi)$ are not empty. Then, we have that
\begin{align*}
I_{\xi}(\varphi+m)&=\inf\left\lbrace I(\xi+k):k\in \mathbb{R}\ \textnormal{s.t.}\ \xi+k\in \BK\ \textnormal{and}\ \xi+(k-m)\geq \varphi\right\rbrace\\
&=\inf\left\lbrace I(\xi+(k-m)+m):k\in \mathbb{R}\ \textnormal{s.t.}\ \xi+k\in \BK\ \textnormal{and}\ \xi+(k-m)\geq \varphi\right\rbrace\\
&\geq \inf\left\lbrace I(\xi+k-m):k\in \mathbb{R}\ \textnormal{s.t.}\ \xi+k\in \BK\ \textnormal{and}\ \xi+(k-m)\geq \varphi\right\rbrace+m\\
&\geq \inf\left\lbrace I(\xi+r):r\in \mathbb{R}\ \textnormal{s.t.}\ \xi+r\in \BK\ \textnormal{and}\ \xi+r\geq \varphi\right\rbrace+m\\
&=I_{\xi}(\varphi)+m
\end{align*}
where the second-to-last inequality follows from the constant superadditivity of $I$ and the fact that since $K$ is convex and $m\geq 0$, the following inequalities
\[
\varphi\leq \xi+k-m\leq \xi+k
\]
imply that $\xi+k-m\in \BK$, whenever $\xi+k\in \BK$. The last inequality follows from observing that for all $k\in \mathbb{R}$ such that $\xi+k\in \BK$ and $\xi+k-m\geq \varphi$, there is $r\in \mathbb{R}$ such that $\xi+r\in \BK$ and $\xi+r\geq \varphi$. The proof of the constant superadditivity of $S_{\xi}$ is analogous and follows the same steps choosing $m\leq 0$ and applying Lemma \ref{banale}.
\end{proof}
Denote by $\mathcal{K}$ the collection of intervals $K$ in $\R$ which satisfy at least one of the following: $K\in\lbrace [0,a], (0,a], [0,a), (0,a), [b,0], (b,0], [b,0),(b,0):a>0,b<0\rbrace$; $ 0\in \textnormal{int}K$; $K=[0,\infty), K=(-\infty,0]$.

\begin{proposition}\label{auxiliary_functionals_positive_sup}
Fix $K\in \mathcal{K}$ and a monotone map $I:\BK\to \bar{\R}$. If $I$ is positively superhomogeneous, then, for all $\xi\in \BK$, $J_{\xi}$ is positively superhomogeneous. If $K=[0,\infty)$, then $H_{\xi}$ is positively superhomogeneous.
\end{proposition}
\begin{proof}
Let $\lambda\geq 1$ and $\varphi,\xi\in \BK$ with $\lambda\varphi\in \BK$. If $J_{\xi}(\lambda\varphi)=\infty$, then the claim follows. If $J_{\xi}(\lambda\varphi)\in \R$, then there exists $\alpha>0$ such that $\alpha\xi\geq \lambda\varphi$ and $\alpha\xi\in \BK$. Since $K\in \mathcal{K}$ and $\lambda\geq 1$, we have that $(\alpha/\lambda)\xi\in \BK$. Thus, $J_{\xi}(\varphi)\in \R$. These simple observations imply that $M_{\xi}(\lambda\varphi)$ and $M_{\xi}(\varphi)$ are not empty. Then, we have that
\begin{align*}
J_{\xi}(\lambda\varphi)&=\inf\left\lbrace I(\alpha\xi):\alpha>0\ \textnormal{s.t.}\ \alpha\xi\geq \lambda \varphi\ \textnormal{and}\ \alpha\xi\in \BK \right\rbrace\\
&=\inf\left\lbrace I\left(\alpha\lambda \frac{\xi}{\lambda}\right):\alpha>0\ \textnormal{s.t.}\ \frac{\alpha}{\lambda} \xi\geq  \varphi\ \textnormal{and}\ \alpha\xi\in \BK  \right\rbrace\\
&\geq \lambda\inf\left\lbrace I\left(\alpha\frac{\xi}{\lambda}\right):\alpha>0\ \textnormal{s.t.}\ \frac{\alpha}{\lambda} \xi\geq  \varphi\ \textnormal{and}\ \alpha\xi\in \BK  \right\rbrace\\
&\geq \lambda\inf\left\lbrace I\left(\gamma\xi\right):\gamma>0\ \textnormal{s.t.}\ \gamma \xi\geq  \varphi\ \textnormal{and}\ \gamma\xi\in \BK \right\rbrace=\lambda J_{\xi}(\varphi)
\end{align*}
where the second-to-last inequality follows from positive superhomogeneity and the fact that $(\alpha/\lambda)\xi\in \BK$ whenever $\alpha\xi\in \BK$. The last inequality follows from the fact that, for each $\alpha>0$, with $(\alpha/\lambda)\xi\geq  \varphi$ and $\alpha\xi\in \BK$, there exists $\gamma>0$ such that $\gamma \xi\geq \varphi$ and $\gamma\xi\in \BK$. Thus, $J_{\xi}$ is positively superhomogeneous. 
\\
If $K=[0,\infty)$, the exact same arguments but with $\lambda\in (0,1)$ such that $\lambda\varphi\in \BK$, yield that $H_{\xi}(\lambda\varphi)\leq \lambda H_{\xi}(\varphi)$, and hence $H_{\xi}$ is positively superhomogeneous. The condition $K=[0,\infty)$ is used to guarantee that $(\alpha/\lambda)\xi\in \BK$.
\end{proof}

\subsubsection{Semicontinuity}

For each $\varphi,\psi\in \B$, the set $[\varphi,\psi]=\lbrace \xi\in \B:\varphi\leq \xi\leq \psi\rbrace$ will be referred to as an \textit{order interval}. A subset $Y$ of $\BK$ is \textit{lower open} (resp., \textit{upper open}) if, for all $\varphi\in Y$, there exists $\varepsilon>0$ such that $[\varphi-\varepsilon,\varphi]\subseteq Y$ (resp., $[\varphi,\varphi+\varepsilon]\subseteq Y$). Fix $K\subseteq \R$. It is important to notice that whenever $K$ is a nontrivial interval, $\BK$ is either lower open, or upper open, or an order interval. Given a function $T:\BK\to \bar{\mathbb{R}}$ we define $T^-:\BK\to \bar{\mathbb{R}}$ as
\[
T^-(\varphi)=\sup\limits_{U\in \mathcal{N}(\varphi, K)}\inf\limits_{\psi\in U}T(\psi)
\]
for all $\varphi\in \BK$, where $\mathcal{N}(\varphi, K)$ denotes the set of all neighborhoods of $\varphi$.\footnote{The neighborhoods of $\varphi\in \BK$ are the subsets of $\BK$ that contain $\varphi$ in their interior, relative to $(\BK,\lVert\cdot\rVert_{\infty})$.} For simplicity, when $K= \R$ we write $\mathcal{N}(\varphi)$. Such function $T^-$ is referred to as the \textit{lower semicontinuous envelope} of $T$. By Proposition 3.5 in \cite{DalMasoBook}, we have that
\begin{equation}\label{aziz_muntazir}
\lbrace T^-\leq k \rbrace = \bigcap\limits_{t>k}\textnormal{cl}(\lbrace T\leq t \rbrace)
\end{equation}
for all $k\in \mathbb{R}$. Analogously, we define the \textit{upper semicontinuous envelope} of $T$ as the function $T^+:\BK\to \bar{\mathbb{R}}$ such that $T^+(\varphi)=\inf_{U\in \mathcal{N}(\varphi, K)}\sup_{\psi\in U}T(\psi)$ for all $\varphi\in \BK$. We report some basic properties of such envelopes, the first two points extend Lemmas 29 and 30 in \cite{cerreia2011uncertainty}.
\begin{lemma}\label{basic_prop_cont_envelopes}
Let $K\subseteq \R$ be convex and fix a map $T:\BK\to \bar{\mathbb{R}}$. Then,
\begin{enumerate}[label=(\roman*)]
\item \label{cont_env_mon_lsc} If $T$ is monotone and $\BK$ is lower open, then $T^-$ is monotone.
\item \label{cont_env_mon_usc} If $T$ is monotone and $\BK$ is either upper open or an order interval, then $T^+$ is monotone.
\item \label{cont_env_quas} If $T$ is quasiconvex (quasiconcave), then $T^-$ and $T^+$ are quasiconvex (quasiconcave).
\item \label{cont_env_cs} If $T$ is constant superadditive, then $T^-$ and $T^+$ are constant superadditive.
\item \label{cont_env_ps} If $T$ is positively superhomogeneous, then $T^-$ and $T^+$ are positively superhomogeneous.
\end{enumerate}
\end{lemma}
\begin{proof}
We show that \textit{\ref{cont_env_mon_lsc}} and \textit{\ref{cont_env_mon_usc}} hold. Lemma 29 and 30 of \cite{cerreia2011uncertainty} can be easily adapted to show that, when $\BK$ is lower open, then $T^{-}(\varphi)=\sup_{n}T(\varphi_n)$ for all $\varphi,\varphi_n\in \BK$ and $\varphi_n\to \varphi$ with $\varphi>\varphi_n$. Analogously, when $\BK$ is upper open, then $T^{+}(\varphi)=\inf_{n}T(\varphi_n)$ for all $\varphi,\varphi_n\in \BK$ and $\varphi_n\to \varphi$ with $\varphi<\varphi_n$. Then, in these cases, monotonicity of $T^-$ and $T^+$ readily follows from the same proofs as in Lemmas 29 and 30 of \cite{cerreia2011uncertainty}.\footnote{Their proofs are provided in the working paper version of their paper.} The only remaining case is the one with $K=[a,b]$ for some $a<b$. To this end, define $\hat{T}:\B\to \bar{\R}$ as
\[
\hat{T}(\varphi)=T(\varphi\wedge b\vee a)
\]
for all $\varphi\in \B$. Clearly, $\hat{T}$ is a monotonic extension of $T$. We want to show that $\hat{T}^+$ is an extension of $T^+$. If this is the case, by the monotonicity of $\hat{T}$, Lemma 29 of \cite{cerreia2011uncertainty} would yield that $T^+$ must be monotone. Fix $\varphi\in \BK$. For all $U\in \mathcal{N}(\varphi)$, we have that
\[
\sup\limits_{\psi\in U}\hat{T}(\psi)\geq \sup\limits_{\psi\in U\cap \BK}\hat{T}(\psi)=\sup\limits_{\psi\in U\cap \BK}T(\psi)
\]
which yields
\begin{align*}
\hat{T}^+(\varphi)&=\inf\limits_{U\in\mathcal{N}(\varphi)}\sup\limits_{\psi\in U}\hat{T}(\psi)\geq \inf\limits_{U\in \mathcal{N}(\varphi)}\sup\limits_{\psi\in U\cap \BK}T(\psi)\\
&=\inf\limits_{O\in \mathcal{N}(\varphi,K)}\sup\limits_{\psi\in O}T(\psi)=T^+(\varphi).
\end{align*}
To prove the converse inequality, notice that, by Proposition 3.6(b) in \cite{DalMasoBook}, there exists a sequence $(\varphi_n)_{n\in \mathbb{N}}$ in $\B$ such that $\varphi_n\to \varphi$ and 
\begin{equation}\label{spacobotilia}
    \hat{T}^+(\varphi)\leq \liminf\limits_{n\to\infty}\hat{T}(\varphi_n).
\end{equation}
Then, by definition of $\hat{T}$, we have that
\begin{align*}
    \hat{T}^+(\varphi)&\leq \liminf\limits_{n\to\infty}\hat{T}(\varphi_n)=\liminf\limits_{n\to\infty}T(\varphi_n\wedge b\vee a)\\
    &\leq \limsup\limits_{n\to\infty}T(\varphi_n\wedge b\vee a)\leq \limsup\limits_{n\to\infty}T^+(\varphi_n\wedge b\vee a)\leq T^{+}(\varphi)
\end{align*}
where the first inequality is due to \eqref{spacobotilia}, the second to definition of $\hat{T}$, the third to the properties of limit superiors, the fourth to the fact that $T^+\geq T$, and the last one to the fact that $\varphi_n\wedge b\vee a\to \varphi$ and that $T^+$ is upper semicontinuous. Thus, we have that $\hat{T}^+(\varphi)= T^+(\varphi)$.
\par\medskip

Point \textit{\ref{cont_env_quas}} follows from \eqref{aziz_muntazir} and the analogous version for $T^+$. To prove point \textit{\ref{cont_env_cs}}, let $\varphi\in \BK$ and $k\geq 0$ with $\varphi+k\in \BK$. Notice that $\mathcal{N}(\varphi+k, K)=\left\lbrace U+k:U\in \mathcal{N}(\varphi, K) \right\rbrace$. Thus, 
    \begin{align*}
T^-(\varphi+k)&=\sup\limits_{U+k\in \mathcal{N}(\varphi, K)}\inf\limits_{\psi\in U+k} T(\psi)=\sup\limits_{U\in \mathcal{N}(\varphi, K)}\inf\limits_{\psi\in U} T(\psi+k)\\
&\geq \sup\limits_{U\in \mathcal{N}(\varphi, K)}\inf\limits_{\psi\in U} T(\psi)+k=T^-(\varphi)+k.
    \end{align*}
    The proof for $T^+$ is analogous. To prove point \textit{\ref{cont_env_ps}} let $\varphi\in \BK$ and $\alpha\geq 1$ such that $\alpha\varphi\in \BK$. Notice that $\mathcal{N}(\alpha\varphi, K)=\left\lbrace \alpha U:U\in \mathcal{N}(\varphi, K) \right\rbrace.$ Thus, 
    \begin{align*}
T^-(\alpha\varphi)&=\sup\limits_{\alpha U\in \mathcal{N}(\varphi, K)}\inf\limits_{\psi\in \alpha U} T(\psi)=\sup\limits_{U\in \mathcal{N}(\varphi, K)}\inf\limits_{\psi\in U} T(\alpha \psi)\\
&\geq \sup\limits_{U\in \mathcal{N}(\varphi, K)}\inf\limits_{\psi\in U} \alpha T(\psi)=\alpha T^-(\varphi). \end{align*} 
The proof for $T^+$ is analogous.\end{proof}

\begin{proposition}\label{cont_representations}
Let $K\subseteq \R$ be convex and $I:\BK\to \bar{\R}$ a monotone map. The following hold
\begin{enumerate}[label=(\roman*)]
\item For all $\varphi\in \BK$, 
\[
I(\varphi)=\min\limits_{\xi\in \BK}I^-_{\xi}(\varphi)=\min\limits_{\xi\in \BK}J^-_{\xi}(\varphi)=\max\limits_{\xi\in \BK}S^-_{\xi}(\varphi)=\max\limits_{\xi\in \BK}H^-_{\xi}(\varphi)
\]
provided that $\BK$ is lower open and $I$ lower semicontinuous.
\item For all $\varphi\in \BK$, 
\[
I(\varphi)=\min\limits_{\xi\in \BK}I^+_{\xi}(\varphi)=\min\limits_{\xi\in \BK}J^+_{\xi}(\varphi)=\max\limits_{\xi\in \BK}S^+_{\xi}(\varphi)=\max\limits_{\xi\in \BK}H^+_{\xi}(\varphi)
\]
provided that $\BK$ is either upper open or an order interval and $I$ upper semicontinuous.
\end{enumerate}
\end{proposition}
\begin{proof}
Let $\varphi\in \BK$. By Proposition \ref{auxiliary_functionals_basic}, we have that 
\[
I^-_{\xi}(\varphi)\geq \sup\limits_{U\in \mathcal{N}(\varphi, K)}\inf\limits_{\psi\in U}I(\psi)=I(\varphi)
\]
for all $\varphi,\xi\in \BK$, where the last equality follows by the continuity of $I$. Moreover, 
\[
I^-_{\varphi}(\varphi)=\sup\limits_{U\in \mathcal{N}(\varphi, K)}\inf\limits_{\psi\in U}I_{\varphi}(\psi)\leq \sup\limits_{U\in \mathcal{N}(\varphi, K)}I_{\varphi}(\varphi)=I(\varphi)
\]
for all $\varphi\in \BK$, where the last equality follows from the fact that $I_{\varphi}(\varphi)=I(\varphi)$. The exact same proof can be provided to retrieve the rest of the equalities. 
\end{proof}

\subsubsection{Quasiconvex duality} \label{quasiconvex_section}
Given $K\subseteq \R$ and $T:\BK\to \bar{\R}$, we define the function $G_T:\R\times\bigtriangleup(S)\to \bar{\R}$ as
\[
G_T(t,p)=\sup\left\lbrace T(\varphi): \varphi\in \BK\ \textnormal{and}\ \int_S\varphi\textnormal{d}p\leq t \right\rbrace
\]
for all $(t,p)\in \R\times\bigtriangleup(S)$. It is immediate to see that $G_T$ is monotone in the first argument. Moreover, by Lemma 31 in \cite{cerreia2011uncertainty}, $G_T$ is quasiconvex.



\begin{proposition}\label{quasic_cs_dual_rep_base}
If $K\subseteq \R$ is convex and unbounded from above and $T:\BK\to \bar{\R}$ is constant superadditive, then $G_{T}$ is constant superadditive in the first argument.
\end{proposition}
\begin{proof}
The proof mimics proof methods provided in \cite{CMMRwealth}. Fix $t\in \R$, $k\geq 0$, and $p\in \bigtriangleup(S)$. By definition of $G_{T}(t,p)$, there exists a sequence $(\varphi_n)_{n\in \mathbb{N}}$ in $B_0(\Sigma,K)$ such that $T(\varphi_n)\uparrow G_{T}(t,p)$ and $\int_S\varphi_n\textnormal{d}p \leq t$ for all $n\in \mathbb{N}$. Then, we have that $\int_S\varphi_n+k\leq t+k$. Since $K$ is unbounded from above, we have that $\varphi_n+k\in \BK$ for all $n\in \mathbb{N}$. Given that $T$ is constant superadditive, we have
\[
G_T(t+k,p)\geq T(\varphi_n+k)\geq T(\varphi_n)+k\to G_T(t,p)+k
\]
proving that $G_T$ is constant superadditive in the first argument.
\end{proof}

\begin{proposition}\label{quasic_ps_dual_rep}
If either $K=(a,\infty)$, $K=[b,\infty)$, or $K=\R$ for some $a>0$ $b\geq 0$, and $T:\BK\to \bar{\R}$ is positively superhomogeneous, then $G_{T}$ is positively superhomogeneous in the first argument. 
\end{proposition}
\begin{proof}
Fix $t\in \R$, $\alpha\geq 1$, and $p\in \bigtriangleup(S)$. By definition of $G_{T}(t,p)$, there exists a sequence $(\varphi_n)_{n\in \mathbb{N}}$ in $\BK$ such that $T(\varphi_n)\uparrow G_{T}(t,p)$ and $\int_S\varphi_n\textnormal{d}p \leq t$ for all $n\in \mathbb{N}$. Then, we have that $\int_S\alpha\varphi_n\leq \alpha t$. By the hypotheses on $K$, we have that $\alpha\varphi_n\in \BK$ for all $n\in \mathbb{N}$. Given that $T$ is positively superhomogeneous, we have
\[
G_T(\alpha t,p)\geq T(\alpha \varphi_n)\geq \alpha T(\varphi_n)\to \alpha G_T(t,p)
\]
proving that $G_T$ is positively superhomogeneous in the first argument.
\end{proof}
The exact same approach also yields the following.
\begin{lemma}\label{quasic_ph_dual_rep}
If $K=[0,\infty)$ and $T:\BK\to \bar{\R}$ is positively homogeneous, then $G_{T}$ is positively homogeneous.
\end{lemma}
\begin{proof}
Fix $t\in \R$, $\alpha>0$, and $p\in \bigtriangleup(S)$. By definition of $G_{T}(t,p)$, there exists a sequence $(\varphi_n)_{n\in \mathbb{N}}$ in $\BK$ such that $T(\varphi_n)\uparrow G_{T}(t,p)$ and $\int_S\varphi_n\textnormal{d}p \leq t$ for all $n\in \mathbb{N}$. Then, we have that $\int_S\alpha\varphi_n\leq \alpha t$. Since $K=[0,\infty)$, we have that $\alpha\varphi_n\in \BK$ for all $n\in \mathbb{N}$. Moreover, since $T$ is positively homogeneous we have that
\[
G_T(\alpha t,p)\geq T(\alpha \varphi_n)=\alpha T(\varphi_n)\to \alpha G_T(t,p)
\]
proving that $G_T$ is positively homogeneous in the first argument.
\end{proof}

\subsubsection{Monotone and normalized functionals}
Section \ref{quasiconvex_section} provided the tools to retrieve the following general representation. 
\begin{proposition}\label{envelop_rep_gen}
Fix a convex $K\subseteq \R$ and a continuous map $I:\BK\to \bar{\R}$. Then, the following are equivalent 
\begin{enumerate}[label=(\roman*)]
\item \label{envgen_1} $I$ is monotone and normalized.
\item \label{envgen_2} There exists a family $\Phi$ of monotone and quasiconvex functionals mapping from $\BK$ to $\bar{\R}$ such that
\[
I(\varphi)=\min\limits_{J\in \Phi}J(\varphi)
\]
for all $\varphi\in \BK$. Moreover, $\min_{J\in \Phi}J(k)=k$ for all $k\in K$. If $\BK$ is lower open, then each $J$ can be taken lower semicontinuous. If $\BK$ is upper open or an order interval, then each $J$ can be taken upper semicontinuous.

\item \label{envgen_3} There exists a family $\Psi$ of monotone and quasiconcave functionals mapping from $\BK$ to $\bar{\R}$ such that
\[
I(\varphi)=\max\limits_{H\in \Psi}H(\varphi)
\]
for all $\varphi\in \BK$. Moreover, $\max_{H\in \Psi}H(k)=k$ for all $k\in K$. If $\BK$ is lower open, then each $H$ can be taken lower semicontinuous. If $\BK$ is upper open or an order interval, then each $H$ can be taken upper semicontinuous.
\end{enumerate}
\end{proposition}
\begin{proof}
\textit{\ref{envgen_1}} implies \textit{\ref{envgen_2}}. Suppose first that $\BK$ is lower open. Let $\Phi=\lbrace I_{\xi}^{-}:\xi\in \BK\rbrace$. By Proposition \ref{auxiliary_functionals_basic} and Lemma \ref{basic_prop_cont_envelopes}, we have that each $J\in \Phi$ is monotone, quasiconvex, and lower semicontinuous. Moreover, by Proposition \ref{cont_representations}, we have that, for all $\varphi\in\BK$,
\[
I(\varphi)=\min\limits_{J\in \Phi}J(\varphi).
\]
Since $I$ is normalized, $\min_{J\in \Phi}J(k)=k$ for all $k\in K$. If $\BK$ is upper open or an order interval the same steps applied to the family $\Phi=\lbrace I^+_{\xi}:\xi\in \BK \rbrace$, yield the claim with upper semicontinuity. 
\par\medskip
\textit{\ref{envgen_1}} implies \textit{\ref{envgen_3}}. We apply the same steps as in the previous implication, but with $\Psi=\lbrace S_{\xi}^{-}:\xi\in \BK\rbrace$, when $B_0(\Sigma,K)$ is lower open, or with $\Psi=\lbrace S_{\xi}^{+}:\xi\in \BK\rbrace$, when $B_0(\Sigma,K)$ is either upper open or an order interval. By Proposition \ref{auxiliary_functionals_basic} and Lemma \ref{basic_prop_cont_envelopes}, we have that each $H\in \Psi$ is monotone, quasiconcave, and lower semicontinuous or upper semicontinuous depending on whether $B_0(\Sigma,K)$ is lower open, upper open, or an order interval. Moreover, by Proposition \ref{cont_representations}, we have that, for all $\varphi\in\BK$,
\[
I(\varphi)=\max\limits_{H\in \Psi}H(\varphi).
\]
Since $I$ is normalized, $\max_{H\in \Psi}H(k)=k$ for all $k\in K$. 
\par\medskip
To conclude, notice that the implications: \textit{\ref{envgen_2}} implies \textit{\ref{envgen_1}} and \textit{\ref{envgen_3}} implies \textit{\ref{envgen_1}} follow immediately from the monotonicity of all $J\in \Phi$ and $H\in \Psi$, and the normalization from $\min_{J\in \Phi}J(k)=k$ and $\max_{H\in \Psi}H(k)=k$ for all $k\in K$.
\end{proof}

\subsubsection{Constant superadditivity}


\begin{proposition}\label{prop_rep}
Fix a convex $K\subseteq \R$ and a continuous map $I:\BK\to \bar{\R}$. Then, the following are equivalent
\begin{enumerate}[label=(\roman*)]
\item \label{item1_prop_rep_sub} $I$ is monotone, normalized, and constant superadditive.
\item \label{item2_prop_rep_sub} There exists a family $\Phi$ of monotone, constant superadditive, and quasiconvex functionals  mapping from $\BK$ to $\bar{\R}$ such that
\[
I(\varphi)=\min\limits_{J\in \Phi}J(\varphi)
\]
for all $\varphi\in B_0(\Sigma,K)$. Moreover,  $\min_{J\in \Phi}J(k)=k$ for all $k\in K$. If $\BK$ is lower open, then each $J$ can be taken lower semicontinuous. If $\BK$ is upper open or an order interval, then each $J$ can be taken upper semicontinuous.

\item \label{item3_prop_rep_sub} There exists a family $\Psi$ of monotone, constant superadditive, and quasiconcave functionals  mapping from $\BK$ to $\bar{\R}$ such that
\[
I(\varphi)=\max\limits_{H\in \Psi}H(\varphi)
\]
for all $\varphi\in B_0(\Sigma,K)$. Moreover, $\max_{H\in \Psi}H(k)=k$ for all $k\in K$. If $\BK$ is lower open, then each $H$ can be taken lower semicontinuous. If $\BK$ is upper open or an order interval, then each $H$ can be taken upper semicontinuous.
\end{enumerate}
\end{proposition}
\begin{proof}
It is immediate to see that both \textit{\ref{item2_prop_rep_sub}} and \textit{\ref{item3_prop_rep_sub}} imply \textit{\ref{item1_prop_rep_sub}}. We prove that \textit{\ref{item1_prop_rep_sub}} implies \textit{\ref{item2_prop_rep_sub}}. By Proposition \ref{cont_representations}, if $K$ is lower open, then
\[
I(\varphi)=\min\limits_{\xi\in \BK}I^-_{\xi}(\varphi)
\]
for all $\varphi\in \BK$. If $K$ is upper open or an order interval, the same result holds with $I^+_{\xi}$ in place of $I^{-}_{\xi}$ for all $\xi\in \BK$. Moreover, by Propositions \ref{auxiliary_functionals_basic}, \ref{auxiliary_functionals_constant_sup}, and Lemma \ref{basic_prop_cont_envelopes} each $I^-_{\xi},I^{+}_{\xi}$, depending on whether $K$ is lower open, upper open, or an order interval, is monotone, quasiconvex, and constant superadditive. Thus, letting either $\Phi=\lbrace I^-_{\xi}:\xi\in \BK\rbrace$ or $\Phi=\lbrace I^+_{\xi}:\xi\in \BK \rbrace$ according to the cases mentioned above, the claim follows. To conclude the proof it is sufficient to show that \textit{\ref{item1_prop_rep_sub}} implies \textit{\ref{item3_prop_rep_sub}}. The proof is totally analogous to the previous steps with either
\[
I(\varphi)=\max\limits_{\xi\in \BK}S^-_{\xi}(\varphi)\ \textnormal{or}\ I(\varphi)=\max\limits_{\xi\in \BK}S^+_{\xi}(\varphi)
\]
in the case where $K$ is lower open, upper open, or an order interval. In particular, we let $\Psi=\lbrace S^-_{\xi}:\xi\in \BK\rbrace$ when $K$ is lower open. Instead if $K$ is upper open or an order interval we set $\Psi=\lbrace S^+_{\xi}:\xi\in \BK\rbrace$. By Propositions \ref{auxiliary_functionals_basic}, \ref{auxiliary_functionals_constant_sup}, and Lemma \ref{basic_prop_cont_envelopes} each $S^-_{\xi},S^{+}_{\xi}$, depending on whether $K$ is lower open, upper open, or an order interval, is monotone, quasiconcave, and constant superadditive. 
\end{proof}

\begin{proposition}\label{quasic_cs_dual_rep} Fix a convex $K\subseteq \R$ and a continuous map $I:\BK\to \R$. Then, the following are equivalent 
\begin{enumerate}[label=(\roman*)]
\item \label{qccsenvgen_1} $I$ is monotone, normalized, and constant superadditive.
\item \label{qccsenvgen_2} There exists a family $\mathcal{G}$ of constant superadditive and increasing in the first argument, and quasiconvex functions $G:\R\times \bigtriangleup(S)\to \R$ such that
\[
I(\varphi)=\max\limits_{G\in \mathcal{G}}\inf\limits_{p\in \bigtriangleup(S)}G\left(\int_S \varphi \textnormal{d}p,p\right)
\]
for all $\varphi\in \BK$. Moreover, $\max_{G\in \mathcal{G}}\inf_{p\in \bigtriangleup(S)}G(k,p)=k$ for all $k\in K$.
\end{enumerate}
\end{proposition}
\begin{proof}
It is immediate to see that \textit{\ref{qccsenvgen_2}} implies \textit{\ref{qccsenvgen_1}}. To prove the converse, notice by Lemma \ref{const_superadd_extension}, $I$ admits a monotone, normalized, and constant superadditive extension $\tilde{I}$. Moreover, we have two cases:
\begin{itemize}
\item If $K$ is lower open, then $\tilde{I}^-$, by the continuity of $I$ and Lemma \ref{basic_prop_cont_envelopes}, is a lower semicontinuous, monotone, and constant superadditive extension of $I$ on $B_0(\Sigma,K_{\infty})$. Then, the result follows from Proposition \ref{prop_rep}, Theorem 36 in \cite{cerreia2011uncertainty},\footnote{Theorem 36 in \cite{cerreia2011uncertainty} is provided for a function mapping in $\R$, but the exact same proof they provide applies to extended real-valued maps.} Proposition \ref{quasic_cs_dual_rep_base}, and the normalization of $I=\tilde{I}^+|_{\BK}$.
\item If $K$ is upper open or an order interval, then $\tilde{I}^+$, by the continuity of $I$ and Lemma \ref{basic_prop_cont_envelopes}, is an upper semicontinuous, monotone, and constant superadditive extension of $I$ on $B_0(\Sigma,K_{\infty})$. Then, the result follows from Proposition \ref{prop_rep}, Theorem 36 in \cite{cerreia2011uncertainty}, Proposition \ref{quasic_cs_dual_rep_base}, and the normalization of $I=\tilde{I}^+|_{\BK}$.
\end{itemize}
Thus, the proof is concluded.
\end{proof}

The following result is adapted from Proposition A.3 in \cite{RuoduCashSubad}. 


\begin{proposition}\label{dependent_variational_abs}
Fix a convex $K\subseteq \R$ and a continuous map $I:\BK\to \R$. Then, the following are equivalent
\begin{enumerate}[label=(\roman*)]
\item \label{item1_prop_rep_sub_var} $I$ is monotone, normalized, and constant superadditive.
\item \label{item2_prop_rep_sub_var} For all $\varphi\in \BK$, there exists a family $\Phi_{\varphi}$ of monotone, constant additive, and convex functionals, mapping from $\BK$ to $\R$, such that
\[
I(\varphi)=\min\limits_{J\in \Phi_{\varphi}}J(\varphi).
\]
Moreover, $\min_{J\in \Phi_{\varphi}}J(k)=k$ for all $k\in K$, and $\Phi_{\varphi_1}\subseteq \Phi_{\varphi_2}$ for all $\varphi_1\geq \varphi_2$ in $\BK$.

\item \label{item3_prop_rep_sub_var} For all $\varphi\in \BK$, there exists a family $\Psi_{\varphi}$ of monotone, constant additive, and concave functionals, mapping from $\BK$ to $\R$, such that
\[
I(\varphi)=\max\limits_{H\in \Psi_{\varphi}}H(\varphi).
\]
Moreover, $\max_{H\in \Psi_{\varphi}}H(k)=k$ for all $k\in K$, and $\Psi_{\varphi_1}\subseteq \Psi_{\varphi_2}$ for all $\varphi_2\geq \varphi_1$ in $\BK$.
\end{enumerate}
\end{proposition}
\begin{proof}
We start proving that \textit{\ref{item1_prop_rep_sub_var}} implies \textit{\ref{item2_prop_rep_sub_var}}. Let $\varphi\in \BK$. For all $\xi\in \BK$ with $\xi\geq \varphi$, we have that for all $s\in S$,
\[
\xi(s)\geq \xi(s)+\sup\{\varphi-\xi\}\geq \xi(s)+\varphi(s)-\xi(s)=\varphi(s).
\]
Thus, since $K$ is convex, for all $\xi\in \BK$ with $\xi\geq \varphi$ it follows that $\xi+\sup\{\varphi-\xi\}\in \BK$ and $\sup\{\varphi-\xi\}\leq 0$. By monotonicity and constant superadditivity of $I$,
\[
I(\varphi)\leq\min\limits_{\xi\geq \varphi}I(\xi+\sup\{\varphi-\xi\})\leq \min\limits_{\xi\geq \varphi}I(\xi)+\sup\{\varphi-\xi\}\leq I(\varphi).
\]
For all $\xi,\psi\in \BK$, define $T_{\xi}(\psi)=I(\xi)+\sup\lbrace\psi-\xi\rbrace$. Therefore, letting $\Phi_{\varphi}=\lbrace T_{\xi} :\xi\in \BK,\ \xi\geq \varphi\rbrace$, we have that for all $\varphi\in \BK$,
\[
I(\varphi)=\min\limits_{J\in \Phi_{\varphi}}J(\varphi).
\]
Notice that each $J\in \Phi_{\varphi}$ is monotone, constant additive, and quasiconvex. Thus, each $J\in \Phi_{\varphi}$ is also convex. Moreover, as $I$ is normalized,  $\min_{J\in \Phi_{k}}J(k)=k$ for all $k\in \mathbb{R}$. Clearly, $\Phi_{\varphi_1}\subseteq \Phi_{\varphi_2}$ for all $\varphi_1\geq \varphi_2$ in $\BK$.

\textit{\ref{item1_prop_rep_sub_var}} implies \textit{\ref{item3_prop_rep_sub_var}}. By Lemma \ref{banalino}, $\bar{I}$ is constant superadditive. Since it is also monotone and normalized, we have that
\[
-I(-\varphi)=\min\limits_{J\in \Phi_{\varphi}}J(\varphi)
\]
for all $\varphi\in \B$ and some family of monotone, constant additive, and convex functionals. Thus, we have that $I(\varphi)=\max_{J\in \Phi_{-\varphi}}\bar{J}(\varphi)$ for all $\varphi\in \B$. Letting $\Psi_{\varphi}=\lbrace \bar{J}:J\in \Phi_{-\varphi}\rbrace$ the claim follows. In particular, notice that
\[
\varphi_1\leq \varphi_2 \Longrightarrow \Psi_{\varphi_1}\subseteq \Psi_{\varphi_2}
\]
for all $\varphi_1,\varphi_2\in \B$.
\par\medskip
\textit{\ref{item3_prop_rep_sub_var}} implies \textit{\ref{item1_prop_rep_sub_var}}.
Monotonicity and normalization are straightforward. We only need to prove that $I$ is constant superadditive. To this end, let $\varphi\in \BK$ and $k\geq 0$ with $\varphi+k\in \BK$. By assumption, we have $\Psi_{\varphi}\subseteq \Psi_{\varphi+k}$, and hence
\[
I(\varphi+k)=\max\limits_{J\in \Psi_{\varphi+k}}J(\varphi+k)=\max\limits_{J\in \Psi_{\varphi+k}}J(\varphi)+k\geq \max\limits_{J\in \Psi_{\varphi}}J(\varphi)+k=I(\varphi)+k.
\]
The proof that \textit{\ref{item2_prop_rep_sub_var}} implies \textit{\ref{item1_prop_rep_sub_var}} is proved analogously upon choosing $k\leq 0$.
\end{proof}
Therefore, applying convex duality results we retrieve the following representation. \cite{RuoduCashSubad} provides a specific form for the penalty functions, the interested reader can consult their Appendix A.

\begin{proposition}\label{dependent_variational}
Fix a convex $K\subseteq \R$ and a continuous map $I:\BK\to \R$. Then, the following are equivalent
\begin{enumerate}[label=(\roman*)]
\item \label{item1_prop_rep_sub_variat_pr} $I$ is monotone, normalized, and constant superadditive.
\item \label{item2_prop_rep_sub_variat_pr} For all $\varphi\in \BK$, there exists a family $C_{\varphi}$ of lower semicontinuous and convex functionals $c:\bigtriangleup(S)\to (-\infty,\infty]$ such that
\[
I(\varphi)=\max\limits_{c\in C_{\varphi}}\min\limits_{p\in \bigtriangleup(S)}\left\lbrace \int_S\varphi\textnormal{d}p+c(p)\right\rbrace.
\]
Moreover, $\max_{c\in C_k}\min_{p\in\bigtriangleup(S)}c(p)=0$ and $C_{\varphi_1}\subseteq C_{\varphi_2}$ for all $\varphi_2 \geq \varphi_1$ in $\BK$.
\end{enumerate}
\end{proposition}
\begin{proof}
The result follows from Proposition \ref{dependent_variational_abs} above and Proposition 4 and Theorem 3 in \cite{Niveloidsextension}.\end{proof}

\subsubsection{Positive superhomogeneity}
Here, we report the representation results for positive superhomogeneity.

\begin{proposition}\label{super_homo_rep}
Fix a convex $K\subseteq \R$ with $\min K=0$ and a continuous map $I:\BK\to \bar{\R}$. Then, the following are equivalent
\begin{enumerate}[label=(\roman*)]
\item \label{item1_sup_homo_rep} $I$ is monotone, normalized, and positively superhomogeneous.
\item \label{item2_sup_homo_rep} There exists a family $\Phi$ of monotone, positively superhomogeneous, upper semicontinuous, and quasiconvex functionals mapping from $\BK$ to $\bar{\R}$ such that
\[
I(\varphi)=\min\limits_{J\in \Phi}J(\varphi)
\]
for all $\varphi\in \BK$. Moreover, $\min_{J\in \Phi}J(k)=k$ for all $k\in K$.

\item \label{item3_sup_homo_rep} There exists a family $\Psi$ of monotone, positively superhomogeneous, upper semicontinuous, and quasiconcave functionals mapping from $\BK$ to $\bar{\R}$ such that
\[
I(\varphi)=\max\limits_{H\in \Psi}H(\varphi)
\]
for all $\varphi\in \BK$. Moreover, $\max_{H\in \Psi}H(k)=k$ for all $k\in K$.
\end{enumerate}
\end{proposition}
\begin{proof}
It is immediate to see that both \textit{\ref{item2_sup_homo_rep}} and \textit{\ref{item3_sup_homo_rep}} imply \textit{\ref{item1_sup_homo_rep}}. If we prove that \textit{\ref{item1_sup_homo_rep}} implies \textit{\ref{item2_sup_homo_rep}} and that \textit{\ref{item1_sup_homo_rep}} implies \textit{\ref{item3_sup_homo_rep}}, then the proof is concluded. We start with \textit{\ref{item1_sup_homo_rep}} implies \textit{\ref{item2_sup_homo_rep}}. First notice that by Lemma \ref{positive_superhomo_extension}, there exists a positively superhomogeneous, monotone, and normalized extension $\tilde{I}$ of $I$ on $B_0(\Sigma,K_{\infty})$. Since $I$ is continuous, and $K$ is either upper open or an order interval, we have that $\tilde{I}^+$ is an upper semicontinuous, positively superhomogeneous, and monotone extension of $I$ on $B_0(\Sigma,K_{\infty})$. 
By Proposition \ref{cont_representations}, since $K_{\infty}$ is upper open and $\tilde{I}^{+}$ is an extension of $I$, we have that
\[
I(\varphi)=\tilde{I}^+(\varphi)=\min\limits_{\xi\in B_0(\Sigma,K_{\infty})}J^+_{\xi}(\varphi)
\]
for all $\varphi\in \BK$, where the $J_{\xi}$'s are the auxiliary functionals associated to $\tilde{I}^+$. By Propositions \ref{auxiliary_functionals_basic}, \ref{auxiliary_functionals_positive_sup}, and Lemma \ref{basic_prop_cont_envelopes}, $J^{+}_{\xi}$ is monotone, positively superhomogeneous, upper semicontinuous, and quasiconvex. Thus, letting $\Phi=\lbrace J^+_{\xi}:\xi\in B_0(\Sigma,K_{\infty}) \rbrace$ according to the cases mentioned above, the claim follows. By the normalization of $I$, we have that $\min_{J\in \Phi}J(k)=\tilde{I}^+(k)=I(k)=k$ for all $k\in K$.
\par\medskip
To prove that \textit{\ref{item1_sup_homo_rep}} implies \textit{\ref{item3_sup_homo_rep}}, we can apply the same exact reasoning to the family $\Psi=\lbrace H^+_{\xi}:\xi\in B_0(\Sigma,K_{\infty}) \rbrace$ of auxiliary functionals associated to the extension $\tilde{I}^+$ of $I$. In particular, by Proposition \ref{cont_representations}, since $K_{\infty}$ is upper open and $\tilde{I}^+$ is an extension of $I$, we have that
\[
I(\varphi)=\tilde{I}^+(\varphi)=\max\limits_{\xi\in B_0(\Sigma,K_{\infty})}H^+_{\xi}(\varphi)
\]
for all $\varphi\in \BK$. By Propositions \ref{auxiliary_functionals_basic}, \ref{auxiliary_functionals_positive_sup}, and Lemma \ref{basic_prop_cont_envelopes} each $H^{+}_{\xi}$ is monotone, upper semicontinuous, quasiconvex, and positively superhomogeneous. Thus, letting $\Psi=\lbrace H^+_{\xi}:\xi\in B_0(\Sigma,K_{\infty}) \rbrace$, the claim follows. By the normalization of $I$, we have that $\max_{H\in \Psi}H(k)=\tilde{I}^+(k)=I(k)=k$ for all $k\in K$.
\end{proof}

\begin{proposition}\label{super_homo_rep_UAP}
Fix a convex $K\subseteq \R$ with $\min K=0$ and a continuous map $I:\BK\to \R$. Then, the following are equivalent
\begin{enumerate}[label=(\roman*)]
\item \label{item1_sup_homo_rep_UAP} $I$ is monotone, normalized, and positively superhomogeneous.
\item \label{item2_sup_homo_rep_UAP} There exists a family $\mathcal{G}$ of positively superhomogeneous and increasing in the first argument, and quasiconvex functions $G:\R\times \bigtriangleup(S)\to \R$ such that
\[
I(\varphi)=\max\limits_{G\in \mathcal{G}}\inf\limits_{p\in \bigtriangleup(S)}G\left(\int_S \varphi \textnormal{d}p,p\right)
\]
for all $\varphi\in \BK$. Moreover, $\max_{G\in \mathcal{G}}\inf_{p\in \bigtriangleup(S)}G(k,p)=k$ for all $k\in K$.
\end{enumerate}
\end{proposition}
\begin{proof}
That \textit{\ref{item2_sup_homo_rep_UAP}} implies \textit{\ref{item1_sup_homo_rep_UAP}} is immediate to see. To prove the converse, notice by Lemma \ref{positive_superhomo_extension}, $I$ admits a monotone, normalized, and positively superhomogeneous extension $\tilde{I}$. Moreover, as $K$ is either upper open or an order interval we have that  $\tilde{I}^+$, by the continuity of $I$ and Lemma \ref{basic_prop_cont_envelopes}, is an upper semicontinuous, monotone, and positively superhomogeneous extension of $I$ on $B_0(\Sigma,K_{\infty})$. Then, the result follows from Proposition \ref{super_homo_rep}, Theorem 36 in \cite{cerreia2011uncertainty},\footnote{Theorem 36 in \cite{cerreia2011uncertainty} is provided for a function mapping in $\R$, but the exact same proof they provide applies to extended real-valued maps.}  Proposition \ref{quasic_ps_dual_rep}, and the normalization of $I=\tilde{I}^+|_{\BK}$.
\end{proof}

\begin{proposition}\label{super_homo_rep_confidence_style_abs}
Fix a convex $K\subseteq \R$ with $\min K=0$ and a continuous map $I:\BK\to \R$. Then, the following are equivalent
\begin{enumerate}[label=(\roman*)]
\item \label{item1_sup_homo_rep_conf} $I$ is monotone, normalized, and positively superhomogeneous.
\item \label{item2_sup_homo_rep_conf} For all $\varphi\in \BK$, there exists a family $\Phi_{\varphi}$ of monotone, sublinear functionals, mapping from $\BK$ to $\R$, such that
\[
I(\varphi)=\min\limits_{J\in \Phi_{\varphi}}J(\varphi).
\]
Moreover, $\min\limits_{J\in \Phi_{\varphi}}J(k)=k$ for all $k\in K$, and $\Phi_{\varphi_1}\subseteq \Phi_{\varphi_2}$ for all $\varphi_1\geq \varphi_2$ in $\BK$.

\item \label{item3_sup_homo_rep_conf} For all $\varphi\in \BK$, there exists a family $\Psi_{\varphi}$ of monotone, superlinear functionals, mapping from $\BK$ to $\R$, such that
\[
I(\varphi)=\max\limits_{H\in \Psi_{\varphi}}H(\varphi).
\]
Moreover, $\max_{H\in \Psi_{\varphi}}H(k)=k$ for all $k\in K$, and $\Psi_{\varphi_1}\subseteq \Psi_{\varphi_2}$ for all $\varphi_2\geq \varphi_1$ in $\BK$.
\end{enumerate}
\end{proposition}
\begin{proof}
We start showing that \textit{\ref{item1_sup_homo_rep_conf}} implies \textit{\ref{item2_sup_homo_rep_conf}}. By Lemma \ref{positive_superhomo_extension} we have that $I$ admits a positively superhomogeneous, and monotone extension $\tilde{I}$ on $B_0(\Sigma,K_{\infty})$. Moreover, since $I$ is continuous, and $K$ is either upper open or an order interval, Lemma \ref{basic_prop_cont_envelopes} yields that $\tilde{I}^+$ is an upper semicontinuous, positively superhomogeneous, and monotone extension of $I$. Now we provide the representation result directly for $\tilde{I}^+$. Notice that here $K_{\infty}=[0,\infty)$. Fix $\varphi\in B_0(\Sigma,K_{\infty})$. For all $\xi\geq \varphi$ in $B_0(\Sigma,K_{\infty})$ with $\xi>0$, we have $\sup\left\lbrace \varphi/\xi \right\rbrace\leq 1$, thus by monotonicity and positively superhomogeneity of $I$,
\small
\[
\tilde{I}^+(\varphi)\leq\min\limits_{\xi>0,\xi\geq \varphi}\tilde{I}^+\left(\xi\sup\left\lbrace\frac{\varphi}{\xi}\right\rbrace\right)\leq \min\limits_{\xi>0,\xi\geq \varphi}\tilde{I}^+(\xi)\sup\left\lbrace\frac{\varphi}{\xi}\right\rbrace\leq \tilde{I}^+\left(\varphi+\frac{1}{n}\right)\sup\left\lbrace \frac{\varphi}{\varphi+\frac{1}{n}}\right\rbrace\leq \tilde{I}^+\left(\varphi+\frac{1}{n}\right)
\]
\normalsize
for all $n\in \mathbb{N}$. By upper semicontinuity of $\tilde{I}^+$ it follows that, for all $\varphi\in B_0(\Sigma,K_{\infty})$,
\[
I(\varphi)=\limsup\limits_{n\to\infty}\tilde{I}^+(\varphi)\leq \limsup\limits_{n\to\infty} \min\limits_{\xi>0,\xi\geq \varphi}\tilde{I}^+(\xi)\sup\left\lbrace\frac{\varphi}{\xi}\right\rbrace\leq \limsup\limits_{n\to\infty}\tilde{I}^+\left(\varphi+\frac{1}{n}\right)\leq I(\varphi)
\]
Therefore, for all $\varphi\in \BK$,
\[
I(\varphi)=\tilde{I}^+(\varphi)=\min\limits_{\xi>0,\xi\geq \varphi}I(\xi)\sup\left\lbrace\frac{\varphi}{\xi}\right\rbrace.
\]
For all $\xi\in B_0(\Sigma,K_{\infty})$ with $\xi>0$, define $T_{\xi}(\psi)= I(\xi)\sup\lbrace\psi/\xi\rbrace$ for all $\psi\in B_0(\Sigma,K_{\infty})$. Therefore, letting $\Phi_{\varphi}=\lbrace T_{\xi}:\xi\in B_0(\Sigma,K_{\infty}),\xi>0,\xi\geq \varphi \rbrace$, we have that
\[
I(\varphi)=\min\limits_{J\in \Phi_{\varphi}}J(\varphi)
\]
for all $\varphi\in \BK$. Since $I(\xi)\geq I(0)=0$ for all $\xi\in \BK$, we have that each $J$ is monotone and sublinear. Clearly, $\Phi_{\varphi_1}\subseteq \Phi_{\varphi_2}$ for all $\varphi_1\geq \varphi_2$ in $\BK$. Since $I$ is normalized, $\min_{J\in \Phi_{\varphi}}J(k)=k$ for all $k\in K$.

Now we prove that \textit{\ref{item1_sup_homo_rep_conf}} implies \textit{\ref{item3_sup_homo_rep_conf}}. We start showing that \textit{\ref{item1_sup_homo_rep_conf}} implies \textit{\ref{item2_sup_homo_rep_conf}}. By Lemma \ref{positive_superhomo_extension} we have that $I$ admits a positively superhomogeneous, and monotone extension $\tilde{I}$ on $B_0(\Sigma,K_{\infty})$. Moreover, since $I$ is continuous, and $K$ is either upper open or an order interval, Lemma \ref{basic_prop_cont_envelopes} yields that $\tilde{I}^+$ is an upper semicontinuous, positively superhomogeneous, and monotone normalized extension of $I$. Let $\varphi\in B_0(\Sigma,K_{\infty})$. If $\varphi(s)=0$ for all $s\in S$, we let $\Psi_{\varphi}=\left\lbrace \mathbf{0}\right\rbrace$ and, since $0\in K$, by the normalization of $I$ we have that $I(\varphi)=\max_{T\in \Psi_{\varphi}}T(\varphi)=0$. Clearly $\mathbf{0}$ is monotone and superlinear. Now suppose that $\varphi(s)>0$ for some $s\in S$. Let $\xi\in B_0(\Sigma,K_{\infty})$ with $\varphi\geq \xi$ and $\xi(s)>0$ for some $s\in S$. We have that for all $s'\in S$,
\[
\varphi(s')\geq \xi(s')\inf\limits_{s:\xi(s)>0}\left\lbrace \frac{\varphi(s)}{\xi(s)} \right\rbrace.
\]
Indeed, if $\xi(s')=0$, then the inequality follows as $\varphi(s')\geq 0$. If $\xi(s')>0$, then
\[
\varphi(s')=\xi(s')\frac{\varphi(s')}{\xi(s')}\geq \xi(s')\inf\limits_{s:\xi(s)>0}\left\lbrace \frac{\varphi(s)}{\xi(s)} \right\rbrace.
\]
Paired with monotonicity and positive superhomogeneity of $\tilde{I}^+$, this inequality yields the following
\[
\tilde{I}^+(\varphi)\geq \max\limits_{\xi\neq \mathbf{0},\varphi\geq \xi}\tilde{I}^+\left(\xi\inf\limits_{s:\xi(s)>0}\left\lbrace \frac{\varphi(s)}{\xi(s)}\right\rbrace\right)\geq \max\limits_{\xi\neq \mathbf{0},\varphi\geq \xi}\tilde{I}^+(\xi)\inf\limits_{s:\xi(s)>0}\left\lbrace \frac{\varphi(s)}{\xi(s)}\right\rbrace\geq \tilde{I}^+(\varphi)
\]
where the second-to-last inequality follows from $\inf_{s:\xi(s)>0}\left\lbrace \varphi(s)/\xi(s)\right\rbrace\geq 1$, for all non-identically zero $\xi\in B_0(\Sigma,K_{\infty})$ with $\varphi\geq \xi$.
For all non-identically zero $\xi\in B_0(\Sigma,K_{\infty})$, define $P_{\xi}(\psi)=I(\xi)\inf_{s:\xi(s)>0}\left\lbrace \psi(s)/\xi(s) \right\rbrace$ for all $\psi\in B_0(\Sigma,K_{\infty})$. Since for all $\xi\in B_0(\Sigma,K_{\infty})$, $\tilde{I}^+(\xi)\geq \tilde{I}^+(0)=I(0)=0$, each $P_{\xi}$ is monotone and superlinear. Therefore, letting $\Psi_{\varphi}=\lbrace P_{\xi}:\xi\in B_0(\Sigma,K_{\infty}),\varphi\geq \xi, \xi\neq \mathbf{0} \rbrace$ for all $\varphi\in \BK$, we have found a family of monotone and superlinear functionals such that
\[
I(\varphi)=\tilde{I}^+(\varphi)=\max\limits_{H\in \Psi_{\varphi}}H(\varphi)
\]
for all $\varphi\in \BK$. Clearly, $\Psi_{\varphi_1}\subseteq \Psi_{\varphi_2}$ for all $\varphi_1\leq \varphi_2$ in $\BK$. Since $I$ is normalized, $\max_{H\in \Psi_{\varphi}}H(k)=k$ for all $k\in K$.

To conclude we prove that \textit{\ref{item2_sup_homo_rep_conf}} implies \textit{\ref{item1_sup_homo_rep_conf}}, and that \textit{\ref{item3_sup_homo_rep_conf}} implies \textit{\ref{item1_sup_homo_rep_conf}}. Suppose that \textit{\ref{item2_sup_homo_rep_conf}} holds. It is immediate to notice that $I$ must be monotone and normalized. Positive superhomogeneity follows from the fact that for all $\varphi\in \BK$ and $\alpha\geq 1$, we have $
\Phi_{\alpha\varphi}\subseteq \Phi_{\varphi}$. Indeed, these inclusions imply that
\[
I(\alpha\varphi)=\min\limits_{J\in \Phi_{\alpha\varphi}}J(\alpha\varphi)=\alpha \min\limits_{J\in \Phi_{\alpha\varphi}}J(\varphi)\geq \alpha \min\limits_{J\in \Phi_{\varphi}}J(\varphi)=\alpha I(\varphi).
\]
The proof of \textit{\ref{item3_sup_homo_rep_conf}} implies \textit{\ref{item1_sup_homo_rep_conf}} is analogous upon observing that $\Psi_{\varphi}\subseteq \Psi_{\alpha\varphi}.$
\end{proof}

\begin{proposition}\label{super_homo_rep_confidence_style}
Fix a convex $K\subseteq \R$ with $\min K=0$ and a continuous map $I:\BK\to \R$. Then, the following are equivalent
\begin{enumerate}[label=(\roman*)]
\item \label{item1_sup_homo_rep_conf_Faro} $I$ is monotone, normalized, and positively superhomogeneous.

\item \label{item2_sup_homo_rep_conf_Faro} For all $\varphi\in \BK$, there exists a set $D_{\varphi}$ of upper semicontinuous and quasiconcave $d:\bigtriangleup(S)\to [0,\infty]$,
\[
I(\varphi)=\max\limits_{d\in D_{\varphi}}\inf\limits_{p\in \bigtriangleup(S)}\frac{\int_{S}\varphi\textnormal{d}p}{d(p)}.
\]
Moreover, $\max_{d\in D_{k}}\inf_{p\in \bigtriangleup(S)}k/d(p)=k$ for all $k\in K$ and $D_{\varphi_1}\subseteq D_{\varphi_2}$ for all $\varphi_2\geq \varphi_1$ in $\BK$.
\end{enumerate}
\end{proposition}
\begin{proof}
By Proposition \ref{super_homo_rep_confidence_style_abs}, for all $\varphi\in \BK$, there exists a family $\Psi_{\varphi}$ of monotone, superlinear functionals, mapping from $\BK$ to $\R$, such that
\[
I(\varphi)=\max\limits_{H\in \Psi_{\varphi}}H(\varphi).
\]
Moreover, $\max_{H\in \Psi_{\varphi}}H(k)=k$ for all $k\in K$, and $\Psi_{\varphi_1}\subseteq \Psi_{\varphi_2}$ for all $\varphi_2\geq \varphi_1$ in $\BK$. Fix $\varphi\in \BK$. Each $H\in \Psi_{\varphi}$ is real-valued, and, being monotone and superlinear, also continuous, and hence, by Theorem 36 in \cite{cerreia2011uncertainty},
\[
H(\psi)=\inf\limits_{p \in \bigtriangleup(S)}G_H\left(\int_S\psi\textnormal{d}p,p\right)
\]
for all $\psi\in \BK$. By positive homogeneity of $H$, Lemma \ref{quasic_ph_dual_rep} yields that $G_H$ is also positively homogeneous in the first entry, and hence
\[
H(\psi)=\inf\limits_{p \in \bigtriangleup(S)}\int_S\psi\textnormal{d}p G_H(1,p)
\]
for all $\psi\in \BK$. Given that $H$ is real-valued, it follows that there must exist $p\in \bigtriangleup(S)$ such that $G_H(1,p)<\infty$. Define $d_H:\bigtriangleup(S)\to [0,\infty]$ as $d_H(p)=1/G_H(1,p)$. By the previous observation, there must exist $p\in \bigtriangleup(S)$ such that $d_H(p)>0$. Since $H\in \Psi_{\varphi}$ is monotone, superlinear, by Lemmas 31 and 32 in \cite{cerreia2011uncertainty}, $G_{H}$ is quasiconvex and lower semicontinuous.\footnote{This claim is slightly more subtle, indeed notice that $\BK$ here is not lower open. Though, since each $H$ is monotone and superlinear, it admits a monotone, lower semicontinuous, and positively homogeneous extension, say $\bar{H}$, over the whole $\B$. The $G_H$ we consider here should actually be $G_{\bar{H}}$ which, by the continuity of $\bar{H}$, is lower semicontinuous by Lemma 32 in \cite{cerreia2011uncertainty}.} It follows that, for all $t\in \mathbb{R}$, 
\[
\lbrace p\in \bigtriangleup(S):d_{H}(p)\leq t \rbrace=\lbrace p\in \bigtriangleup(S): G_{H}(1,p)\geq 1/t \rbrace
\] are closed and convex for all $t\geq 0$. Thus, $d_{H}$ is quasiconcave and upper semicontinuous. Letting $D_{\varphi}=\lbrace d_H:H\in \Psi_{\varphi}\rbrace$, we have that
\[
I(\varphi)=\max\limits_{d\in D_{\varphi}}\inf\limits_{p\in \bigtriangleup(S)}\frac{\int_{S}\varphi\textnormal{d}p}{d(p)}.
\]
By the normalization of $I$ we have $\max_{d\in D_{k}}\min_{p\in \bigtriangleup(S)}k/d(p)=k$ for all $k\in K$. Moreover,  $D_{\varphi_1}\subseteq D_{\varphi_2}$ for all $\varphi_2\geq \varphi_1$ in $\BK$, because for all such $\varphi_1,\varphi_2\in \BK$, we have $\Psi_{\varphi_1}\subseteq \Psi_{\varphi_2}$. The converse is proved using the same exact steps of \textit{\ref{item3_sup_homo_rep_conf}} implies \textit{\ref{item1_sup_homo_rep_conf}} in Proposition \ref{super_homo_rep_confidence_style_abs}.
\end{proof}

\subsection{Proofs of the results in the main text} \label{app_B}

\begin{proof}[Proof of Lemma \ref{Lemma_1}]
The if part is straightforward. Therefore, we prove the only if part. The existence of an affine $u$ follows from Theorem 8 in \cite{HerMiln53}. Moreover, by Proposition 1 in \cite{OmnibusSimone} there exists a continuous, monotone, and normalized functional $I:B_0(\Sigma.u(X))\to \R$ such that $I\circ u$ represents $\succsim$. The rest of the claim follows from Proposition \ref{envelop_rep_gen}.
\end{proof}

\subsubsection{Absolute ambiguity attitudes}

The proofs of the results concerning absolute ambiguity attitudes are based on the following result which generalizes Lemma 3 in \cite{Xuechanging}. For any subset $K\subseteq \R$, denote by $\textnormal{int}K$ the interior of $K$.

\begin{lemma}\label{lem:DAAA_pref}
Let $\succsim$ be a binary relation over $\mathcal{F}$. The following are equivalent
\begin{enumerate}[label=(\roman*)]
    \item $\succsim$ is MBA, and exhibits \nameref{ax_DAAA}.
    \item There exist an affine function $u:X\to \mathbb{R}$ and a continuous, monotone, normalized, and constant superadditive functional $I:B_0(\Sigma,u(X))\to \mathbb{R}$ such that,  for all $f,g\in \mathcal{F}$,
    \[
    f\succsim g \Longleftrightarrow I(u(f))\geq I(u(g)).
    \]
\end{enumerate}
\end{lemma}
\begin{proof}
The if part is straightforward. We prove the only if part. In particular, the existence of an affine $u$ follows from Theorem 8 in \cite{HerMiln53}. Since $u$ is unique up-to positive affine transformations, it is without loss of generality to assume that $0\in \textnormal{int}u(X)$ and denote by $x_0\in X$ the element such that $u(x_0)=0$. Moreover, by Proposition 1 in \cite{OmnibusSimone} there exists a continuous, monotone, and normalized functional $I:B_0(\Sigma,u(X))\to \R$ such that $I\circ u$ represents $\succsim$. We now show the constant superadditivity of $I$. Since $I\circ u$ represents $\succsim$, \nameref{ax_DAAA} yields
\begin{equation}\label{eq:ammazzotutti}
I(\alpha \varphi+(1-\alpha)u(x_0))\geq u(\alpha x+(1-\alpha)x_0) \Longrightarrow I(\alpha \varphi+(1-\alpha)k)\geq \alpha u(x)+(1-\alpha)k
\end{equation}
for all $\alpha\in (0,1)$, $k\in u(X)$ with $k\geq 0$, $\varphi\in B_0(\Sigma,u(X))$, and $x\in X$. Suppose that, for some $\alpha\in (0,1)$, $k\in u(X)\cap \mathbb{R}_{+}$, and $\varphi\in B_0(\Sigma,u(X))$, we have that
\begin{equation}\label{ref_ineq_proof}
I(\alpha \varphi+(1-\alpha)k)<I(\alpha \varphi)+(1-\alpha)k.
\end{equation}
Notice that since $I$ is monotone and normalized and $u(X)$ is convex, $I(\alpha \varphi)\in \alpha u(X)$ as $\alpha \varphi(s)\in \alpha u(X)$ for all $s\in S$. Therefore, there exists $x\in X$, such that $I(\alpha \varphi)=\alpha u(x)$ and, since $u$ is affine, $I(\alpha\varphi)=u(\alpha x+(1-\alpha)x_0)$. Then, by \eqref{eq:ammazzotutti} and \eqref{ref_ineq_proof}
\[
I(\alpha \varphi+(1-\alpha)u(x_0))<u(\alpha x+(1-\alpha)x_0)=I(\alpha\varphi)=I(\alpha \varphi+(1-\alpha)u(x_0))
\]
a contradiction. Thus, by Lemma \ref{setticemia} in Appendix \ref{app_A}, $I$ is constant superadditive.
\end{proof}


\begin{proof}[Proof of Theorem \ref{thm:UAPstyleDAAA}]
The only if part follows by Lemma \ref{lem:DAAA_pref} and Proposition \ref{quasic_cs_dual_rep} in Appendix \ref{app_A}. For the if part define $I:B_0(\Sigma,u(X))\to \R$ as
\[
I(\varphi)=\max_{G\in \mathcal{G}}\inf_{p\in \bigtriangleup(S)}G\left(\int\varphi\textnormal{d}p,p\right)
\]
for all $\varphi\in \BK$. Since $\mathcal{G}$ is linearly continuous, we have that $I$ is continuous. Moreover, $I$ is normalized as $\max_{G\in \mathcal{G}}\inf_{p\in \bigtriangleup(S)}G(k,p)$ for all $k\in u(X)$. Since each $G\in \mathcal{G}$ is monotone and constant superadditive in the first argument, we have that $I$ is monotone and constant superadditive. Thus, by Lemma \ref{lem:DAAA_pref}, $\succsim$ is an MBA preference relation exhibiting decreasing absolute ambiguity aversion.
\end{proof}
\begin{proof}[Proof of Proposition \ref{dependent_variational_main}]It follows by Lemma \ref{lem:DAAA_pref} and Proposition \ref{dependent_variational} in Appendix \ref{app_A}.
\end{proof}

\subsubsection{Relative ambiguity attitudes}

\begin{lemma}\label{lem:IRAA_pref}
Let $\succsim$ be a binary relation over $\mathcal{F}$. The following are equivalent
\begin{enumerate}[label=(\roman*)]
    \item $\succsim$ is MBA, admits a worst consequence, and exhibits \nameref{ax_DRAA}.
    \item There exist an affine function $u:X\to \mathbb{R}$ with $0=\min u(X)$, and a continuous, monotone, normalized, and positively superhomogeneous functional $I:B_0(\Sigma,u(X))\to \mathbb{R}$ such that, for all $f,g\in \mathcal{F}$,
    \[
    f\succsim g \Longleftrightarrow I(u(f))\geq I(u(g)).
    \]
\end{enumerate}
\end{lemma}
\begin{proof}
We start with the if part. Since $0=\min u(X)$ and $I$ is continuous, monotone, and normalized we have that $\succsim$ is an MBA preference admitting a worst consequence, say $x_*$. Therefore, we prove that $\succsim$ satisfies \nameref{ax_DRAA}. To this end suppose that $\alpha f+(1-\alpha)x_*\succsim \alpha y+(1-\alpha)x_*$ for some $\alpha\in(0,1]$, $f\in \mathcal{F}$, and $y\in X$. Then, since $u(x_*)=0$, $u$ is affine, and $I$ is normalized we have $I(\alpha u(f))\geq \alpha u(y)$. Since $I$ is positively superhomogeneous, we have that $\gamma\mapsto I(\gamma u(f))/\gamma$ is increasing in $(0,1]$, and hence for all $\beta\in [\alpha,1]$ we have $I(\beta u(f))\geq \beta u(y)$. This yields
\[
I(u(\beta f+(1-\beta)x_*))=I(\beta u(f))\geq \beta u(y)=u(\beta y+(1-\beta)x_*)
\]
and hence $\beta f+(1-\beta)x_*\succsim \beta y+(1-\beta)x_*$.
\par\medskip
We now prove the only if part. Analogously to the proofs of Lemmas \ref{Lemma_1} and \ref{lem:DAAA_pref} there exists an affine $u$ and a continuous, monotone, and normalized functional $I:B_0(\Sigma,u(X))\to \R$ such that $I\circ u$ represents $\succsim$. We are left to show the positively superhomogeneity of $I$. First of all, since $I\circ u$ represents $\succsim$, \nameref{ax_DRAA} implies that
\begin{equation}\label{eq:IRAA}
I(\alpha \varphi+(1-\alpha)u(x_*))\geq u(\alpha x+(1-\alpha)x_*) \Longrightarrow I(\beta \varphi+(1-\beta)u(x_*))\geq u(\beta x+(1-\beta)x_*)
\end{equation}
for all $\alpha,\beta\in (0,1]$ with $\alpha\leq \beta$, $\varphi\in B_0(\Sigma,u(X))$, $x\in X$. By the affinity of $u$, \eqref{eq:IRAA} can be rewritten as
\[
I(\alpha \varphi)\geq \alpha k \Longrightarrow I(\beta \varphi)\geq \beta k
\]
for all $\alpha,\beta\in (0,1]$ with $\alpha\leq \beta$ and $k\in u(X)$. Since $I(B_0(\Sigma,u(X)))\subseteq u(X)$, it follows that $
\gamma\mapsto I(\gamma \varphi)/\gamma$
is increasing in $(0,1]$, and hence
\[
I(\gamma\varphi)\leq \gamma I(\varphi)
\]
for all $\gamma \in (0,1]$ and $\varphi\in B_0(\Sigma,u(X))$, showing the statement.
\end{proof}


\begin{proof}[Proof of Theorem \ref{representation_IRAA}]
The only if part follows by Lemma \ref{lem:IRAA_pref} and Proposition \ref{super_homo_rep_UAP} in Appendix \ref{app_A}. The if part is analogous to the proof of Theorem \ref{thm:UAPstyleDAAA} substituting constant superadditivity with positive superhomogeneity.
\end{proof}

\begin{proof}[Proof of Proposition \ref{dependent_confidence_main}]
It follows by Lemma \ref{lem:IRAA_pref} and Proposition \ref{super_homo_rep_confidence_style} in Appendix \ref{app_A}.
\end{proof}

\subsubsection{Examples \ref{example1DAAA} and \ref{example1DRAA}} \label{Appendix_B3}

We first report a fairly standard preliminary result. The points 1 and 2 below are due to \cite{MarinacciMontrucchioUnique} (see Appendix C therein). 

\begin{lemma}\label{Lemma_constsuppossup_examples}
 If $K\subseteq \R$ is convex and $\phi:K\to \R$ is strictly increasing, continuous, concave, then $I_q:B_0(\Sigma,K)\to \R$ is continuous, normalized, monotone, and quasiconcave for $q\in Q$. Moreover
\begin{enumerate}[label=(\roman*)]
    \item\label{DARAsoeu} If $K=[0,\infty)$ and $\phi$ is twice differentiable in $(0,\infty)$ and DARA, then $I_q$ is constant superadditive.
    \item\label{DRRAsoeu} If $K=[0,\infty)$ and $\phi$ is twice differentiable in $(0,\infty)$ and DRRA, then $I_q$ is positively superhomogeneous. 
\end{enumerate}
\end{lemma}
\begin{proof}
Since $\phi$ is strictly increasing, monotonicity and normalization are immediate. As per continuity, suppose that $\psi_n\to \psi$. Then, since $\phi$ is continuous we have that $\phi(\psi_n)\to \phi(\psi)$ pointwise, and by the monotonicity of $\phi$, we have that $|\phi(\psi_n)|\leq \phi(\lVert\psi_n \rVert_{\infty})$ for all $n\in \mathbb{N}$. Since $\lVert \psi_n\rVert_{\infty}\to \lVert \psi\rVert_{\infty}$, we have that there exists $K>0$ such that $\sup_n\lVert \psi_n\rVert_{\infty}\leq K$. Thus, for all $n\in \mathbb{N}$, we have that $|\phi(\psi_n)|\leq \phi(K)$. Then, by the dominated convergence theorem, we have that 
\[
\int \phi(\psi_n)\textnormal{d}q\to \int \phi(\psi)\textnormal{d}q.
\]
Since, $\phi^{-1}$ is also continuous, it follows that $I_q(\psi_n)\to I_q(\psi)$, proving the continuity. Moreover, since $\phi^{-1}$ is strictly increasing and $\phi$ is concave, we have that
\begin{align*}
I_q(\alpha\varphi+(1-\alpha)\psi)&=\phi^{-1}\left(\int\phi(\alpha\varphi+(1-\alpha)\psi)\textnormal{d}q\right)\geq \phi^{-1}\left(\int\alpha\phi(\varphi)+(1-\alpha)\phi(\psi)\textnormal{d}q\right)\\
&\geq \min\left\lbrace \phi^{-1}\left(\int\phi(\varphi)\textnormal{d}q\right),\phi^{-1}\left(\int\phi(\psi)\textnormal{d}q\right)  \right\rbrace= \min\left\lbrace I_q(\varphi),I_q(\psi) \right\rbrace
\end{align*}
for all $\alpha\in [0,1]$, and $\varphi,\psi\in\BK$, where the last inequality follows from the fact that being monotone $\phi^{-1}$ is also quasiconcave.

Points \textit{\ref{DARAsoeu}} and \textit{\ref{DRRAsoeu}} follow from \cite{MarinacciMontrucchioUnique} Theorem 12 and Corollary 1 (see also Proposition 6 therein), respectively. Their proofs are presented in a finite state space, though this is not necessary for the argument.\footnote{The details are available upon request.} This concludes the proof.
\end{proof}
We say that a function $G:\R\times \bigtriangleup(S)\to \bar{\R}$ is \textit{linearly continuous} if the map
\[
\varphi\mapsto \inf\limits_{p\in \bigtriangleup(S)}G\left(\int \varphi \textnormal{d}p,p\right)
\]
is continuous.
\begin{lemma}\label{representation_for_examples}
For all $q\in Q$, there exists a linearly continuous $G_q:\R\times \bigtriangleup(S)\to \bar{\R}$ monotone in the first entry, quasiconvex, and such that for all $\varphi\in \BK$
\[
I_q(\varphi)=\inf\limits_{p\in \bigtriangleup(S)}G_q\left(\int \varphi \textnormal{d}p,p\right).
\]
Moreover,
\begin{enumerate}[label=(\roman*)]
\item If $\phi$ is DARA, then $G_q$ is constant superadditive in the first argument.
\item If $\phi$ is DRRA, then $G_q$ is positively superhomogeneous in the first argument.
\end{enumerate}
\end{lemma}
\begin{proof}
By Lemma \ref{Lemma_constsuppossup_examples}, we have that each $I_q$ is continuous, normalized, monotone, and quasiconcave. Therefore, by Theorem 36 in \cite{cerreia2011uncertainty}, we have that for all $\varphi\in \BK$,
\[
I_q(\varphi)=\inf\limits_{p\in \bigtriangleup(S)}G_q\left(\int \varphi \textnormal{d}p,p\right)
\]
where $G_q(t,p)=\sup\left\lbrace I_q(\varphi):\varphi\in \BK\ \textnormal{and}\ \int \varphi \textnormal{d}p\leq t \right\rbrace$ for all $(t,p)\in \R\times \bigtriangleup(S)$. Following the results in \cite{cerreia2011uncertainty} we have that $G_q$ is monotone in the first entry, quasiconvex, and such that for all $k\in K$, $\inf_{p\in \bigtriangleup(S)}G_q(k,p)=k$. Since $I_q$ is continuous, we have that $G_q$ is linearly continuous. Moreover,
\begin{enumerate}[label=\textit{(\roman*)}]
\item \label{item_1phiDARA} If $\phi$ is DARA, by Lemma \ref{Lemma_constsuppossup_examples}, we have that $I_q$ is constant superadditive, and hence by Proposition \ref{quasic_cs_dual_rep_base}, we have that $G_q$ is constant superadditive.
\item \label{item_2phiDRRA} If $\phi$ is DRRA, by Lemma \ref{Lemma_constsuppossup_examples}, we have that $I_q$ is positively superhomogeneous, and hence by Proposition \ref{quasic_ps_dual_rep}, we have that $G_q$ is positively superhomogeneous.
\end{enumerate}
These conclude the proof.
\end{proof}
By Lemma \ref{Lemma_constsuppossup_examples} and \ref{representation_for_examples}, MBA preferences $\succsim$ represented by $f\mapsto \max_{q\in Q}I_q(u(f))$ satisfy: \textit{\ref{item_1phiDARA}} if $\phi$ is DARA, then $\succsim$ satisfies the conditons of Theorem \ref{thm:UAPstyleDAAA}; \textit{\ref{item_2phiDRRA}} if $\phi$ is DRRA, then $\succsim$ satisfies the conditions of Theorem \ref{representation_IRAA}. We provide two examples of functions $\phi:[0,\infty)\to \R$ that satisfy these hypotheses:
\begin{itemize}
\item $\phi(t)=\sqrt{t}$ for all $t\geq 0$, is strictly increasing, continuous, concave, and DARA. To see the latter, notice that $$-\frac{\phi''(t)}{\phi'(t)}=\frac{1}{2t},$$
a (strictly) decreasing function.
\item $\phi(t)=\sqrt{t}+t$ for all $t\geq 0$, is strictly increasing, continuous, concave, and DRRA. To see the latter, notice that, for all $t>0$,
\[
-t\frac{\phi''(t)}{\phi'(t)}=\frac{1}{4\sqrt{t}+2}
\]
which is decreasing in $t$.
\end{itemize}


\subsubsection{Remark \ref{concave_DRRA_smooth}} \label{appendix_smooth_DRRA}

\begin{proposition} \label{prop_smooth_draa} 
Let $\phi:[1,\infty)\to \R$ strictly increasing, twice differentiable on $(1,\infty)$, and concave and suppose that $\Sigma$ is non-trivial. The functional $I_{\phi,\mu}:B_0(\Sigma,[1,\infty))\to \R$ defined as
\[
I_{\phi,\mu}(\varphi)=\phi^{-1}\left(\int\phi\Big(\int \varphi \textnormal{d}p\Big) \textnormal{d}\mu\right)
\]
for all $\varphi\in B_0(\Sigma,[1,\infty))$
is positively superhomogeneous for all  countably additive probabilities $\mu$ over $\bigtriangleup^{\sigma}(S)$ if and only if $\phi$ is DRRA.
\end{proposition}
\begin{proof} Fix a countable additive $\mu$ over $\bigtriangleup^{\sigma}(S)$. First we prove that $I_{\phi,\mu}$ is positively superhomogeneous if and only if $\hat{I}_{\varphi,\mu}$ defined as $$\hat{I}_{\phi,\mu}:\varphi\mapsto \log\left(\phi^{-1}\Big(\int\phi\Big(\int e^{\varphi} \textnormal{d} p \Big) \textnormal{d} \mu\Big)\right)$$ is constant superadditive. First suppose that $I_{\phi,\mu}$ is positively superhomogeneous. For all $\varphi\in B_0(\Sigma,[0,\infty))$ and $k\geq 0$, we have that
\begin{align*} \allowdisplaybreaks
\hat{I}_{\phi,\mu}(\varphi+k)&=\log\left(\phi^{-1}\Big(\int\phi\Big(\int e^{\varphi+k} \textnormal{d} p \Big) \textnormal{d} \mu\Big)\right)=\log\left(\phi^{-1}\Big(\int\phi\Big(\int e^{\varphi}e^{k} \textnormal{d} p \Big) \textnormal{d} \mu\Big)\right)\\
&\geq\log \left(e^{k}\phi^{-1}\Big(\int\phi\Big(\int e^{\varphi} \textnormal{d} p \Big) \textnormal{d} \mu\Big)\right)
=\log\left(\phi^{-1}\Big(\int\phi\Big(\int e^{\varphi} \textnormal{d} p \Big) \textnormal{d} \mu\Big)\right)+k=\hat{I}_{\phi,\mu}(\varphi)+k.
\end{align*}
Therefore, $\hat{I}_{\phi,\mu}$ is constant superadditive. Conversely, suppose that $\hat{I}_{\phi,\mu}$ is constant superadditive. Then, we have that for all $\alpha\geq 1$ and $\varphi\in B_0(\Sigma,[1,\infty))$, 
\begin{align*}
I_{\phi,\mu}(\alpha\varphi)=e^{\hat{I}_{\phi,\mu}(\log(\alpha\varphi))}=e^{\hat{I}_{\phi,\mu}(\log(\alpha)+\log(\varphi))}\geq e^{\hat{I}_{\phi,\mu}(\log(\varphi))}e^{\log(\alpha)}=\alpha e^{\hat{I}_{\phi,\mu}(\log(\varphi))}=\alpha I_{\phi,\mu}(\varphi)
\end{align*}
Therefore, $I_{\phi,\mu}$ is positively superhomogeneous. In particular, notice that $\hat{I}_{\phi,\mu}=I_{\hat{\phi},\mu}$ with $\hat{\phi}(t)=\phi(e^t)$ for all $t\geq 0$. Notice that 
\[
\hat{\phi}'(t)=e^t\phi'(e^t)\ \textnormal{and}\ \hat{\phi}''(t)=e^t\phi'(e^t)+e^{2t}\phi''(e^t),
\]
which implies
\begin{align*}
\hat{\phi}''(t)\leq 0 \Leftrightarrow \phi'(e^t)\leq -e^t\phi''(e^t)\Leftrightarrow 0\leq -e^t \frac{\phi''(e^t)}{\phi'(e^t)}
\Leftrightarrow 0\leq-\frac{\phi''(e^t)}{\phi'(e^t)}
\end{align*}
Since $\phi$ is strictly increasing, we have $\phi'(e^t)>0$; since $\phi$ is concave, we have $\phi''(e^t)\leq 0$, and hence $\hat{\phi}''(t)\leq 0$ for all $t$. Therefore, $\hat{\phi}$ is strictly increasing and concave. Moreover, by Proposition 7 in \cite{Xuechanging},\footnote{Notice that in the proof \cite{Xuechanging} uses the full unboundedness requirement only for the preferential characterization of constant superadditivity in terms of decreasing absolute ambiguity aversion. The rest of the proof, and in particular the cited results from \cite{cerreia2011uncertainty}, do not depend on the full unboundedness condition.} $\hat{I}_{\phi,\mu}$ is constant superadditive for all countable additive probabilities $\mu$ in $\bigtriangleup^{\sigma}(S)$ if and only if $\hat{\phi}$ is DARA. In turn, $\hat{\phi}$ is DARA if and only if $\phi$ is DRRA. Therefore, connecting all the proved equivalences, $I_{\phi,\mu}$ is positively superhomogeneous if and only if $\phi$ is DRRA.
\end{proof}

\subsubsection{Risk sharing} \label{appendix_risk_sharing}
\begin{lemma}\label{lemma:combined}
Let $\succsim$ be a binary relation over $\mathcal{F}$. The following are equivalent
\begin{enumerate}[label=(\roman*)]
    \item $\succsim$ is MBA, admits a worst consequence, and exhibits \nameref{ax_DAAA} and increasing relative ambiguity aversion.
    \item There exist an affine function $u:X\to \mathbb{R}$ with $\min u(X)=0$ and a continuous, monotone, normalized, constant superadditive, and positively subhomogeneous functional $I:B_0(\Sigma,u(X))\to \mathbb{R}$ such that,  for all $f,g\in \mathcal{F}$,
    \[
    f\succsim g \Longleftrightarrow I(u(f))\geq I(u(g)).
    \]
\end{enumerate}
Moreover, the same statement holds if we require $\succsim$ to satisfy \nameref{ax_nc_DAAA} and \nameref{ax_nc_IRAA}, along with strict constant superadditivity and strict positively subhomogeneity.
\end{lemma}
\begin{proof}
The equivalence between the two points follows immediately by combining Lemma \ref{lem:DAAA_pref} and Lemma \ref{lem:IRAA_pref}, adapting the latter specularly to increasing relative ambiguity aversion. Therefore, we only need to prove that non-constant decreasing absolute ambiguity aversion is equivalent to the strict constant superadditivity of $I$, and the analogous statement for non-constant increasing relative ambiguity aversion and the strict positive subhomogeneity of $I$. The proofs of such equivalences are omitted for the sake of brevity as they are completely analogous (up to changing few weak signs with their strict counterparts) to the proofs of the lemmas mentioned above.\footnote{The details are available upon request.}
\end{proof}

\begin{proof}[Proof of Proposition \ref{risksharingth}]
Let $x\in \R_{++}$ and $g\neq x$ in $\R^S_+$ with $V(g)\geq V(x)$. Then, for all $\lambda\in [0,1]$ it holds that
 $$V(\lambda g+(1-\lambda)x)\geq V(\lambda g)+(1-\lambda)x\geq \lambda V(g)+(1-\lambda)x\geq \lambda V(x)+(1-\lambda)V(x)=V(x),$$
with at least one of the inequalities being strict by assumption. 
Since $V$ is nice, the result follows by Proposition 2 (points (2) and (3)) in \cite{ghirardatorisk}.
\end{proof}

\begin{proof}[Proof of Corollary \ref{corollary_NiceAllocations}]
    The result follows by a direct application of  Theorem 3 by \cite{ghirardatorisk} and Proposition \ref{risksharingth}.
\end{proof}

\bibliographystyle{ecta}
\bibliography{bibliography}



\end{document}